\documentclass[journal]{IEEEtran}

\usepackage{graphicx,subfigure}
\usepackage{amsmath,amssymb,stmaryrd,mathtools,amsthm, mathtools}
\usepackage{flushend}
\usepackage{cite}
\usepackage{array}

\usepackage{algorithm,algorithmicx,algpseudocode}

\DeclareMathAlphabet{\mathcal}{OMS}{cmsy}{m}{n}

\usepackage{acronym}
\usepackage{hyperref}
\usepackage{url}

\usepackage{tikz}
\usetikzlibrary{positioning,arrows,chains,scopes,calc,fit,backgrounds}

\newtheorem{theorem}{Theorem}
\newtheorem{corollary}{Corollary}

\theoremstyle{definition}
\newtheorem{definition}{Definition}
\newtheorem{problem}{Problem}

\theoremstyle{remark}
\newtheorem{remark}{Remark}

\DeclareMathOperator{\F}{\rotatebox[origin=c]{45}{$\Box$}}
\DeclareMathOperator{\G}{\Box}
\DeclareMathOperator{\X}{\bigcirc}
\DeclareMathOperator{\until}{\mathcal{U}}

\DeclareMathOperator{\llangle}{\langle\langle}
\DeclareMathOperator{\rrangle}{\rangle\rangle}
\DeclarePairedDelimiter{\abs}{\lvert}{\rvert}

\DeclareMathOperator*{\argmin}{argmin}
\graphicspath{{./figures/}}

\newcommand{\ignore}[1]{}

\def \ACE {\mathrm{ACE} }
\def \anc {\mathrm{anc} }

\def \f {\varphi}
\def \r {\rho} 
\def \x {\mathbf{x}} 
\def \u {\mathbf{u}} 
\def \w {\mathbf{w}} 
\def \y {\mathbf{y}} 

\def\X {{\cal X}}
\def\U {{\cal U}}
\def\W {{\cal W}}

\def\Reals {\mathbb{R}}

\acrodef{TCL}{Thermostatically Controlled Load}
\acrodef{AC}{Air Conditioning}
\acrodef{FERC}{Federal Energy Regulatory Commission}
\acrodef{AGC}{Automatic Generation Control}
\acrodef{CAISO}{California Independent System Operator}
\acrodef{PJM}{Pennsylvania-New Jersey-Maryland}
\acrodef{MCP}{Market Clearing Price}
\acrodef{NGR}{Non-Generator Resource}
\acrodef{SoC}{State of Charge}
\acrodef{RES}{Renewable Energy Source}
\acrodef{DR}{Demand Response}
\acrodef{DLC}{Direct Load Control}
\acrodef{EV}{Electric Vehicle}
\acrodef{SDSP}{Supply-Duct Static Pressure}
\acrodef{SISO}{Single-Input Single-Output}
\acrodef{ARX}{Auto Regressive with eXogenous input}
\acrodef{FERC}{Federal Energy Regulatory Commission}
\acrodef{AGC}{automatic generation control}
\acrodef{LFC}{Load-Frequency Control}
\acrodef{ACE}{Area Control Error}
\acrodef{AFRC}{Area Frequency Response Characteristic}
\acrodef{RHC}{receding horizon control}
\acrodef{HVAC}{heating, ventilation and air-conditioning}
\acrodef{BMC}{bounded model checking}
\acrodef{STL}{Signal temporal logic}
\acrodef{LTL}{linear temporal logic}
\acrodef{LP}{linear program}
\acrodef{MTL}{metric temporal logic}
\acrodef{MILP}{mixed integer linear program}
\acrodef{MPC}{Model predictive control}
\acrodef{RHC}{Receding horizon control}

\def \occ {\text{occ}}
\def \comf {\text{comf}}
\def \f {\varphi}

\def\dt{ \text{\small $\Delta$}t}

\begin{document}

\title{Model Predictive Control for Signal Temporal Logic Specifications}

\author{Vasumathi Raman$^1$, Alexandre Donz\'{e}$^2$, Mehdi Maasoumy$^3$,\\ Richard
  M. Murray$^1$, Alberto Sangiovanni-Vincentelli$^2$ and Sanjit A. Seshia$^2$% <-this % stops a space
\thanks{This work was supported in part by TerraSwarm, one of six centers of STARnet, a Semiconductor Research Corporation program sponsored by MARCO and DARPA.}% <-this % stops a space
\thanks{$^1$V. Raman and R. M. Murray are with the California Institute of Technology, Pasadena, CA 91125, USA
        {\tt\small vasu@caltech.edu, murray@cds.caltech.edu}}%
\thanks{$^2$A. Donz\'{e}, A. Sangiovanni-Vincentelli and S. A. Seshia are with the Department of Electrical Engineering and Computer Science, UC Berkeley, Berkeley, CA 94720, USA
        {\tt\small donze@berkeley.edu, alberto@berkeley.edu, sseshia@eecs.berkeley.edu}}%
\thanks{$^3$M. Maasoumy is with C3 Energy Inc. Redwood City, CA 94063, USA
        {\tt\small mehdi.maasoumy@c3energy.com}}

%\thanks{V. Raman is with the United Technologies Research Center, Berkeley, CA 94705, USA
%        {\tt\small vasumathi.raman@gmail.com}}%
%\thanks{A. Donz\'{e}, M. Maasoumy, A. Sangiovanni-Vincentelli and S. A. Seshia are with the Department of Electrical Engineering and Computer Science, UC Berkeley, Berkeley, CA 94720, USA
%        {\tt\small donze@berkeley.edu, maasoumy@eecs.berkeley.edu, alberto@berkeley.edu, sseshia@eecs.berkeley.edu}}%
%\thanks{R. M. Murray is with the California Institute of Technology, Pasadena, CA 91125, USA
%        {\tt\small murray@cds.caltech.edu}}%
\thanks{Manuscript received January 14, 2016.}}

\maketitle

\begin{abstract}
We present a mathematical programming-based method for model predictive control of cyber-physical systems subject to signal temporal logic (STL) specifications. We describe the use of STL to specify a wide range of properties of these systems, including safety, response and bounded liveness. For synthesis, we encode STL specifications as mixed integer-linear constraints on the system variables in the optimization problem at each step of a receding horizon control framework. We prove correctness of our algorithms, and present experimental results for controller synthesis for building energy and climate control.

%%% Local Variables:
%%% mode: latex
%%% TeX-master: "root_tac_part1"
%%% End:

\end{abstract}


\IEEEpeerreviewmaketitle

\section{Introduction}
Controlling a cyber-physical system (CPS) involves handling complex interactions between computing
components and their physical environment, and often necessitates hierarchies of
controllers. Typically at the highest level, a \emph{supervisory} controller is responsible for
making high-level decisions, while at the lowest level traditional control laws such as PID control
are used. In general, the design of these different controllers is done mostly in isolation at each
level, and their combination is implemented ad hoc. As the complexity of these systems
grows, reasoning about the correctness of interactions between the various layers of control becomes
increasingly challenging, begging automation.

Formal methods is the subfield of computer science concerned with verification and synthesis, i.e.,
automatic and rigorous design of digital systems. It provides mathematical
 formalisms for specifying behaviors and algorithms for verification and synthesis of
a system against properties specified within these formalisms.  Methods for synthesis
of correct-by-construction discrete supervisory controllers have been developed and successfully
used for cyber-physical systems in domains including robotics \cite{Fainekos06} and aircraft power system design
\cite{Nuzzo13}. However, for physical systems that require constraints not just on the order of
events, but on the temporal distance between them, simulation and testing is still the method of
choice for validating properties and establishing guarantees; the exact exhaustive verification 
or synthesis of such systems is in general undecidable \cite{AlurHLP00}.

%\ac{RHC}, also called 
\ac{MPC} or receding horizon control (RHC) is based on iterative, finite horizon optimization
over a model of the \emph{plant}, i.e. the system to be controlled.
At any given time $t$, the current plant state is observed, and
an optimal control strategy computed for some finite time horizon in the future, $[t,t+H]$. An
online calculation is performed to explore trajectories originating from the
current state, and an optimal control strategy computed up to time $t+H$. To provide robustness with
respect to modelling errors, only the first step of the computed optimal control strategy is
implemented. The plant state is then sampled again, and new calculations are performed on a horizon
of $H$ starting from the new current state. While the global optimality of such a receding horizon
approach is not ensured, it tends to do well in practice: in addition to reducing computational
complexity, it improves the system robustness with respect to exogenous disturbances and modeling
uncertainties \cite{MurrayHJMPDF02}. Another reason \ac{MPC} is particularly attractive to industry
is its ability to handle constrained dynamical systems \cite{MayneRRS00}.

\ac{STL}~\cite{MalerN04} was originally developed in order to specify and monitor the expected
behavior of physical systems, including temporal constraints between events. \ac{STL} allows the
specification of properties of dense-time, real-valued signals, and the automatic generation of
monitors for testing these properties on individual simulation traces. It has since been applied to
the analysis of several types of continuous and hybrid systems, including dynamical systems and
analog circuits, where the continuous variables represent quantities like currents and voltages in
a circuit. STL has the advantage of naturally admitting \emph{quantitative} semantics which, in
addition to the yes/no answer to the satisfaction question, provide a real number that grades the
quality of the satisfaction or violation. Such semantics have been defined for timed logics,
including \ac{MTL} \cite{FainekosP09} and \ac{STL} \cite{DonzeM10}, to assess the \emph{robustness}
of systems to parameter or timing variations. 

%
%An additional advantage of \ac{STL} is that, in
%contrast to standard semantics for other temporal logics, the satisfaction of an \ac{STL} formula
%with bounded modalities has unambiguous interpretation with respect to a finite signal or sequence
%of states.

%%%

In this paper, %the first of a two-part series, 
we solve the problem of control synthesis from \ac{STL} specifications, using a
receding horizon approach. We allow the user to specify desired properties of the system using an \ac{STL}
formula, and synthesize control such that the system satisfies that specification, while using a
receding horizon approach to ensure practicality and robustness.  We do so by decomposing the
\ac{STL} specifications into a series of formulas over each time horizon, such that synthesizing a
controller fulfilling the formula at each horizon results in satisfaction of the global
specification.  Recent work on optimal control synthesis of aircraft load management
systems~\cite{MaasoumyHOLMS} represented \ac{STL}-like specifications as time-dependent equality and
inequality constraints, yielding a \ac{MILP}. The \ac{MILP} was then solved in an \ac{MPC} framework,
yielding an optimal control policy. However, the manual transformation of specifications into
equality and inequality constraints is cumbersome and problem-specific. As a key contribution, this
paper presents two {\em automatically-generated} \ac{MILP} encodings for such \ac{STL}
specifications.

Our main contribution is a pair of bounded model checking-style encodings \cite{BiereHJLS06} for \ac{STL} specifications as \ac{MILP}
constraints on a cyber-physical system. We show how these encodings can be used to generate
open-loop control signals that satisfy finite and infinite horizon \ac{STL} properties and, moreover,
to generate signals that maximize quantitative (robust) satisfaction. 
We provide a fragment of STL, denoted SNN-STL, such that, under
reasonable assumptions on the system dynamics, the problem
of synthesizing an open-loop control sequence such that the
system satisfies a provided specification is a \ac{LP}, and therefore polynomial-time solvable.
We also demonstrate how our
\ac{MILP} formulation of the \ac{STL} synthesis problem can be used in an \ac{MPC} framework to
compute feasible and optimal controllers for cyber-physical systems under timed specifications. We
present experimental results comparing both encodings, and two case studies: one on a thermal
model of an \ac{HVAC} system, and another in the context of regulation services in a micro-grid. These case
studies were previously reported in \cite{Cyphy14,cdc14}. We show how the \ac{MPC} schemes in these
examples can be framed in terms of synthesis from an \ac{STL} specification, and present simulation
results to illustrate the effectiveness of our methodology.

% As a case study, we consider simplified models of a smart building-level micro-grid with uncertain
% demand and generation presented in~\cite{MaasoumySSP14}. We build on the hierarchical control
% framework introduced in \cite{Maasoumy13}, in which \ac{RHC} is implemented on top of the existing
% state-of-the-art \ac{AGC} for exploiting the demand-side flexibility of a commercial building, in
% order to provide fast frequency regulation services to the power grid. We then show how the
% \ac{RHC} scheme for controlling the ancillary service power flow from such buildings can be framed
% in terms of synthesis from a Signal Temporal Logic (\ac{STL}) specification. Preliminary
% simulation results illustrate the effectiveness of the proposed methodology for grid frequency
% regulation.

%%% Local Variables:
%%% mode: latex
%%% TeX-master: "root_tac"
%%% End:

\section{Preliminaries}
\subsection{Systems}

We consider a continuous-time system $\Sigma$ of the form
$$
\dot{x}_t = f(x_t,u_t,w_t)
%x(t) = f(x,u,w)
$$
where $x_t \in {\cal X} \subseteq (\mathbb{R}^{n_c} \times \{0,1\}^{n_l})$ are the continuous and
binary/logical states, $u_t \in {U} \subseteq (\mathbb{R}^{m_c} \times \{0,1\}^{m_l})$ are the
(continuous and logical) control inputs, $w_t \in {W} \subseteq (\mathbb{R}^{e_c} \times
\{0,1\}^{e_l})$ are the external environment inputs (also referred to as ``disturbances''). As $x_t$ includes binary components, the above ODE may
contain discontinuities corresponding to switches in these components. In general, such hybrid 
systems are more accurately modeled using differential-algebraic equations, but we do not dwell
on this point since we approximate their behavior with a difference equation, as follows.

%Note that as $x$ includes binary components, the function $f$ implies a continuous abstraction 
%(in the form of an ordinary differential equation) for an inherently discrete model. 
%Techniques such as those described in \cite{BradleyA12} can be applied to construct such an abstraction. 

Given a sampling time $\dt>0$, we assume that $\Sigma$ admits a discrete-time approximation
$\Sigma_d$ of the form
\begin{equation} x(t_{k+1}) = f_d(x(t_k), u(t_k), w(t_k))\label{eq:sys_dyn}
\end{equation} where for all $k>0$, $t_{k+1}-t_k = \dt$. A \emph{run} of $\Sigma_d$ is a sequence
$$
\xi = (x_0u_0w_0)(x_1u_1w_1)(x_2u_2w_2)...
$$
where $x_k=x(t_k) \in {\cal X}$ is the state of the system at index $k$, and for each $k\in \mathbb{N}$, $u_k= u(t_k) \in {U}$, $w_k=w(t_k) \in {W}$ and $x_{k+1} =f_d(x_k,u_{k},w_{k})$. We assume that given an initial state $x_0 \in X$, a control input sequence ${\bf u}^N = u_0u_1u_2 \hdots u_{N-1} \in U^N$, and a sequence of environment inputs ${\bf w}^N = w_0w_1w_2 \hdots w_{N-1} \in W^N$, the resulting horizon-$N$ run of a system modeled by equation (\ref{eq:sys_dyn}), which we denote by
$$\xi(x_0, {\bf u}^N, {\bf
  w}^N)=(x_0u_0w_0)(x_1u_1w_1)(x_2u_2w_2)...(x_Nu_Nw_N),
$$ is unique.  In addition, we introduce a
generic cost function $J(\xi(x_0,{\bf u},{\bf w}))$ that maps (infinite and finite) runs to
$\mathbb{R}$.

%%% Local Variables:
%%% mode: latex
%%% TeX-master: "root_tac_part1"
%%% End:

\subsection{Signal Temporal Logic}

We consider STL formulas defined recursively according to the following grammar:
$$
\f ::= \pi^\mu \mid \neg \psi \mid \f_1 \land \f_2 \mid \F_{[a,b]}~\f \mid \f_1~\until_{[a,b]}~\f_2,
$$
where $\pi^\mu$ is an atomic predicate $\X\times\U\times\W \rightarrow \mathbb{B}$ whose truth value is
determined by the sign of a function $\mu: \X\times\U\times\W \rightarrow \Reals$ and $\f_1, \f_2$ are STL
formulas. The fact that a run $\xi(x_0, \u,\w)$ satisfies an STL formula $\f$ is denoted by
$\xi \models \f$. Informally, $\xi \models \F_{[a,b]} \f$ if
$\f$ holds at some time step between $a$ and $b$, and
$\xi \models \f~\until_{[a,b]}~\psi$ if $\f$ holds at every time step before $\psi$
holds, and $\psi$ holds at some time step between $a$ and $b$. Additionally, we define
$\G_{[a,b]}\f = \neg\F_{[a,b]}(\neg \f)$, so that $\xi \models \G_{[a,b]} \f$ if
$\f$ holds at \emph{all} times between $a$ and $b$. Formally, the validity of a formula $\f$ with
respect to the run $\xi$ is defined inductively as follows
\[
\begin{array}{lll}
\xi \models \f & \Leftrightarrow& (\xi,t_0) \models \f\\
(\xi,t_k) \models\pi^\mu &\Leftrightarrow& \mu(x_k,y_k,u_k,w_k) > 0\\
(\xi,t_k) \models \neg \psi &\Leftrightarrow& \neg(\xi,t_k) \models \psi)\\
(\xi,t_k) \models \f \land \psi &\Leftrightarrow& (\xi,t_k) \models \f \land (\xi,t_k) \models \psi\\
(\xi,t_k) \models \F_{[a,b]} \f &\Leftrightarrow& \exists t_{k'}\in [t_k\!+\!a, t_k\!+\!b], (\xi,t_{k'}) \models \f\\
 (\xi,t_k) \models \f\until_{[a,b]}\psi\!\!\!\! &\Leftrightarrow& \exists t_{k'} \in [t_k\!+\!a,t_k\!+\!b] \mbox{ s.t. } (\xi,t_{k'}) \models \psi \\
&&\land \forall t_{k''} \in [t_{k},t_{k'}], (\xi,t_{k''}) \models \f.
\end{array}
\]

An STL formula $\f$ is \emph{bounded-time} if it contains no unbounded operators; the
\emph{bound} of $\f$ is the maximum over the sums of all nested upper bounds on the temporal
operators, and provides a conservative maximum trajectory length required to decide its
satisfiability. For example, for $\G_{[0,10]} \F_{[1,6]} \f$, a trajectory of length $N$ such that 
$t_N \ge 10 + 6 = 16$ is sufficient to determine whether the formula is satisfiable. 

\begin{remark}
  Here we have defined a semantics for STL over discrete-time signals, which is formally equivalent to the
  simpler \ac{LTL}, once time and predicates are abstracted into steps and
  Boolean variables, respectively. There are several advantages of still using STL over LTL, though. First, STL
  allows us to explicitly use real time in our specifications instead of abstract integer indices,
  which improves the readability relative to the original system's behaviors. Second, although in
  the rest of this paper we focus on the control of the discrete-time system $\Sigma_d$, our
  goal is to use the resulting controller for the control of the continuous system $\Sigma$. Hence
  the specifications should be independent from the sampling time $\dt$. Finally, note that the
  relationship between the continuous-time and discrete-time semantics of STL, depending on
  discretization error and sampling time, is beyond the scope of this paper. The interested reader
  can refer to \cite{fainekos07} for further discussion on this topic.
  \end{remark}

\subsection{Quantitative semantics for STL}\label{robust_defn}

Quantitative or robust semantics for STL are defined by providing 
a real-valued function $\r^\f$ of signal $\xi$ and time $t$ such
that $\r^\f(\xi,t) > 0 \Rightarrow (\xi,t) \models \f$. We define 
one such function recursively, as follows:
\[
\begin{array}{lll}
\r^{\pi^\mu}(\xi,t_k)&=& \mu(x_k,y_k,u_k,w_k)\\
\r^{\neg \psi}(\xi,t_k)&=&  -\r^\psi(\xi, t_k) \\
\r^{\f_1 \land \f_2}(\xi,t_k)&=& \min (\r^{\f_1}(\xi,t_k),\r^{\f_2}(\xi,t_k)    )\\
\r^{\f_1 \lor \f_2}(\xi,t_k)&=& \max (\r^{\f_1}(\xi,t_k),\r^{\f_2}(\xi,t_k)    )\\
\r^{\F_{[a,b]} \psi}(\xi,t_k)&=& \max_{t_{k'}\in [t+a, t+b]}\r^{\psi}(\xi,t_{k'})\\
\r^{\f_1\until_{[a,b]}\f_2}(\xi,t_k) \!\!\!\!\!\!&=&  \max_{t_{k'}\in [t+a, t+b]} ( \min (\r^{\f_2}(\xi,t_{k'}),    \\
&&~~~~~~~~~~~~~~\min_{t_{k''} \in [t_k,t_{k'}]} \r^{\f_1}(\xi,t_{k''}))
\end{array}
\]

Note that if $\mu(x_k,y_k,u_k,w_k) = 0$, neither $\r^{\pi^\mu}(\xi,t_k) > 0$ nor $\r^{\pi^\mu}(\xi,t_k) > 0$. 
Therefore, $(\xi,t) \models \f \not \Rightarrow \r^\f(\xi,t) > 0$.
To simplify notation, we denote $\r^{\pi^\mu}$ by $\r^{\mu}$ for the remainder of the paper. The
robustness of satisfaction for an arbitrary STL formula is computed recursively from the above
semantics by propagating the values of the functions associated with
each operand using $\min$ and $\max$ operators corresponding to the various STL operators. For
example, the robust satisfaction of $\pi^{\mu_1}$ where $\mu_1(x) = x-3>0$ at time $0$ is
$\r^{\mu_1}(\xi,0) = x_0-3$. The robust satisfaction of $\mu_1 \wedge \mu_2$ is the minimum of
$\r^{\mu_1}$ and $\r^{\mu_2}$. Temporal operators are treated as conjunctions and disjunctions along
the time axis: since we deal with discrete time, the robustness of satisfaction of
$\f= \G_{[0,2.1]} \mu_1$ is
\[
\begin{array}{lll}
\r^\f(x,0) &=& \min_{t_k \in [0,2.1]} \r^{\mu_1}(x,t_k)\\
&=& \min \{x_0-3, x_1-3, \hdots, x_K-3\},
  \end{array}
\]
where  $0 \leq t_0 < t_1 < \hdots < t_K \leq 2.1 < t_{K+1} $.

The robustness score $\r^\f(\xi,t)$ can be interpreted as \emph{how much} $\xi$ satisfies
$\f$. Its absolute value can be viewed as the signed distance of $\xi$ from the set of
trajectories satisfying or violating $\f$, in the space of projections with respect to the
function $\mu$ that define the predicates of $\f$ \cite{FainekosP09}.

%%% Local Variables:
%%% mode: latex
%%% TeX-master: "root_tac_part1"
%%% End:

\section{Problem Statement}
\noindent We now formally state the STL control synthesis problem and its model predictive control formulation.
Given an STL formula $\f$, a \emph{cost function} of the form $J(x_0,\u,\w,\f) \in \Reals$, an
initial state $x_0 \in \X$, a horizon $L$ and a reference disturbance signal $\w \in \W^N$, we
formulate two problems: \emph{open-loop} and \emph{closed-loop} synthesis.  The two scenarios are
depicted as block diagrams in Fig.~\ref{fig:open} and Fig.~\ref{fig:closed}. 

\begin{problem}[open-loop]\label{prob:open}
Compute $\u^*=u_0^*u_1^*\hdots u_{N-1}^*$ where
$$
\begin{array}{lll}
\u^*=& \displaystyle \argmin_{\u \in \U^N} J(x_0, \u, \w, \f)\\
&\mbox{s.t. } \xi(x_0, \u, \w) \models \varphi
\end{array}
$$
\end{problem}

Note that we assume that the state of the plant is fully observable, and the environment inputs are known in advance.

\begin{problem}[closed-loop]\label{prob:mpc}
Given a horizon $0<L<N$, for all $0 \leq k \leq N-L$, compute $u_k^*=u_k^{L*}$, the first element of the
sequence $\u_k^{L*}=u_k^{L*}u_{k+1}^{L*}\hdots u_{k+L-1}^{L*}$ satisfying 
$$
\begin{array}{lll}
\u_k^{L*}=& \displaystyle \argmin_{\u^L_K \in \U^L} J(x_k, \u_k^L, \w_k, \f)\\
&\mbox{s.t. } \xi(x_k, \u^L_k, \w_k) \models \varphi
\end{array}
$$
\end{problem}

The closed-loop formulation corresponds to a model predictive control scheme, where the reference disturbance can change at each iteration $k$.

In Sections~\ref{openloop} and \ref{mpc_synth}, we present both an open-loop solution to Problem
\ref{prob:open}, and a solution to Problem~\ref{prob:mpc} for a large class of STL formulas. In the
absence of an objective function $J$ on runs of the system, we maximize the robustness of the
generated runs with respect to $\varphi$.  A key component of our solution is encoding the \ac{STL}
specifications as \ac{MILP} constraints, which can be combined with \ac{MILP} constraints
representing the system dynamics to efficiently solve the resulting state-constrained
optimization problem.

%%% Local Variables: 
%%% mode: latex
%%% TeX-master: "root_tac_part1"
%%% End: 

\section{Open-loop Controller Synthesis}\label{openloop}

To solve Problem \ref{prob:open}, we extend the bounded model checking encoding of \cite{BiereCCZ99} from finite, discrete systems to dynamical systems using mixed-integer programming instead of SAT. Our presentation and notation below follow that of \cite{WolffTM14}. For open-loop controller synthesis, we will search for a trajectory of length $N$ that satisfies $\varphi$. To admit \ac{STL} formulas describing infinite runs, we parametrize an infinite sequence of states using a finite sequence with a loop. Imposing this lasso-shaped structure renders our synthesis procedure conservative for infinite-state systems, in the sense that a solution may exist that is not found when imposing such a structure. For finite-state systems, the lasso shape is without loss of generality but we must still find an appropriate trajectory length $N$. 

\begin{figure}[h]
\centering
\begin{tikzpicture}[node distance=.6cm, >=stealth',every join/.style=->,
box/.style={draw, text centered, text width=1.6cm, minimum height=1cm},
nobox/.style={text centered, text width=1.8cm},
transform shape,scale=1.0
]  
\node (in) [nobox] {$\f,\Sigma_d,J$};
%\node (milp) [box, below=1cm of in] {Parametric MILP};
\node (synthesis) [box, below=1cm of in,fill=black!5] {Synthesis};
%\node (solver) [box, right=of milp,text width=1.5cm] {MILP Solver};
%\node (sol) [nobox,right=of solver] {$\u*=$ $u^*_0u^*_1\hdots u^*_{N-1}$};
\node (sol) [nobox,right=of synthesis] {$\u*=$ $u^*_0u^*_1\hdots u^*_{N-1}$};
\node (plant) [box, right= of sol, text width=1.5cm,fill=black!5] {Plant $\Sigma$};
\node (trace) [right= of plant] {$\x,\y$};  
%\node (x0wd) [nobox, below= 0.6cm of milp] {$x_0$, $w_0,\hdots,w_{N-1}$};
\node (x0wd) [nobox, below= 0.6cm of synthesis] {$x_0$, $w_0,\hdots,w_{N-1}$};
\node (x0w) [nobox, below= 0.6cm of plant] {$x_0$, $\w$};
%\draw[->] (in) -- node [left] {synthesis} (milp);
\draw[->] (in) -- node [left] {} (synthesis);
%\draw[->] (x0wd) -- (milp);  
\draw[->] (x0wd) -- (synthesis);  
\draw[->] (x0w) -- (plant);  
{[start chain]
%\chainin (milp) [join];
\chainin (synthesis) [join];
%\chainin (solver) [join];
\chainin (sol) [join];
\chainin (plant) [join];
\chainin (trace) [join];
}
\end{tikzpicture}
\caption{Open-loop problem formulation: given an STL formula $\phi$, 
the discrete-time plant model $\Sigma_d$, the cost function $J$, the goal is to
%parametric Mixed-Integer Linear Program 
generate a sequence of control inputs $\u^*$ over a horizon of $N$ time steps.
The problem is additionally parametrized by the the initial state $x_0$ and disturbance vector $\w$.
%The parameters of this MILP are the initial state $x_0$ and disturbance vector $\w$. When those are provided, a
%solver can compute an optimal solution $\u^*$ for horizon $N$ which is passed and used by the
%plant $\Sigma$.
}  
\label{fig:open}
\end{figure}
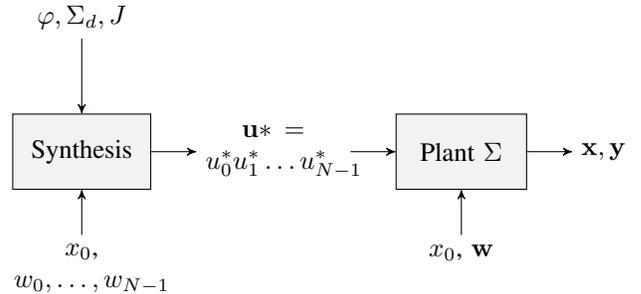

The encoding of Problem \ref{prob:open} as an \ac{MILP} consists of
system constraints, loop constraints and STL constraints, as defined below. 

\subsection{Constraints on system evolution}\label{encode:sys_dyn}

The first component of the set of constraints is provided by the system
model. Our approach applies to any system that yields to a MILP formulation for model predictive control over horizon $N$.
The system constraints encode valid finite (horizon-$N$) trajectories for a system with form
(\ref{eq:sys_dyn}) -- these constraints hold if and only if the trajectory ${\bf x}(x_0, {\bf u}_N)$
satisfies (\ref{eq:sys_dyn}) for $t = 0,1,...,N$. Note that this is quite general, and accommodates
any system for which the resulting constraints and objectives form a mixed integer-linear
program. 
An example is the smart grid regulation control system presented in
\cite{MaasoumySSP14}. Other useful examples include mixed logical dynamical systems such as those
presented %by Bemporad and Morari
in \cite{BemporadM99}. Other cost functions and system dynamics can also be included by using appropriate solvers.

\subsection{Loop constraints for trajectory parametrization}\label{encode:lasso}
As described above, to synthesize open-loop control for unbounded (infinite-horizon) specifications, we parametrize the trajectory as a lasso, i.e. constrain it to contain a loop. This loop encoding is again inspired by the basic idea of bounded model checking \cite{BiereHJLS06}, which is to consider only a finite prefix of a path when looking for a solution to an existential model checking problem. A crucial observation is that although the considered path prefix is finite, it can still represent an infinite path if there is a loop back from the last state to any of the previous states. 

To enforce the existence of a loop in the finite system trajectory, we introduce $N$ binary variables $l_1, ..., l_N$, which determine where the loop forms. These are constrained such that only one can be high at a time, and if $l_j = 1$, then $x_{j-1} = x_N$. The following constraints enforce these requirements:

\begin{itemize}
\item $\sum_{j=1}^N l_j = 1$
\item $x_N \le x_{j-1} + M_j(1-l_j), ~j = 1,...,N$,
\item $x_N \ge x_{j-1} + M_j(1-l_j), ~j = 1,...,N$,
\end{itemize} 

\noindent where $M_j$ are sufficiently large positive numbers, picked based on $\cal X$. 

\subsection{Boolean encoding of STL constraints}\label{sec:milp}

Given a formula $\varphi$, we introduce a variable $z^\varphi_t$, whose value is tied to a set of
mixed integer linear constraints required for the satisfaction of $\varphi$ at position $t$ in the
state sequence of horizon $N$. In other words, $z^\varphi_t$ has an associated set of MILP constraints such that
$z^\varphi_t = 1$ if and only if $\varphi$ holds at position $t$. We recursively generate the MILP
constraints corresponding to $z^\varphi_0$ -- the value of this variable determines whether a
formula $\varphi$ holds in the initial state.

\subsubsection{Predicates} \label{encode:preds} The predicates are represented by constraints on
system state variables. For each predicate $\mu \in P$, we introduce binary variables $z^\mu_t \in
\{0,1\}$ for times $t = 0,1,...,N$. The following constraints enforce that $z^\mu_t = 1$ if and only
if $\mu(x_t) > 0$:
\[
\begin{array}{lll}
\mu(x_t) &\le&  M_tz^\mu_t - \epsilon_t\\
-\mu(x_t) &\le&  M_t(1-z^\mu_t) - \epsilon_t
\end{array}
\]
where $M_t$ are sufficiently large positive numbers, and $\epsilon_t$ are sufficiently small positive numbers that serve to bound $\mu(x_t)$ away from $0$. %Note that $z_t = 1$ if and only if $\mu(x_t) > 0$. 
This encoding restricts the set of STL formulas that can be encoded using our approach to those over linear predicates, but admits arbitrary STL formulas over such predicates.

\subsubsection{Boolean operations on MILP variables}\label{encode:bool}
%Here we follow the example of \cite{KaramanSF08} when encoding negation, conjunction and disjunction of variables using mixed integer-linear constraints. 
As described in Section \ref{encode:preds}, each predicate $\mu$ has an associated binary variable
$z^\mu_t$ which equals 1 if $\mu$ holds at time $t$, and 0 otherwise. In fact, by the recursive
definition of our MILP constraints on STL formulas, each operand $\varphi$ in a Boolean operation
has a corresponding variable $z^\varphi_t$ which is 1 if $\varphi$ holds at $t$ and 0 otherwise. Here
we define Boolean operations on these variables: these are the building blocks of our recursive
encoding. The definitions in this subsection are consistent with those in \cite{WolffTM14}.

Logical operations on variables $z^\psi_t \in [0,1]$ are defined as follows:

\noindent \underline{Negation: $z^\psi_t = \neg z^\varphi_t$}~~~~~~~~~~~~~
$
z^\psi_t = 1 - z^\varphi_t
$\\

\noindent \underline{Conjunction: $z^\psi_t = \displaystyle\land_{i=1}^m z^{\varphi_i}_{t_i}$}
$
\begin{array}{l}
z^\psi_t \le z^{\varphi_i}_{t_i}, i = 1,...,m,\\
z^\psi_t \ge 1 - m + \sum_{i=1}^mz^{\varphi_i}_{t_i}
\end{array}
$\\\\

\noindent \underline{Disjunction: $z^\psi_t = \displaystyle\lor_{i=1}^m z^{\varphi_i}_{t_i}$}~~~
$
\begin{array}{l}
z^\psi_t \ge z^{\varphi_i}_{t_i}, i = 1,...,m,\\
z^\psi_t \le \sum_{i=1}^mz^{\varphi_i}_{t_i}
\end{array}
$\\

Given a formula $\psi$ containing a Boolean operation, we add new continuous variables $z^\psi_t \in [0,1]$, 
%to represent its truth value at each time step of the parametrized trajectory. 
and set $z^\psi_t =  \neg z^\mu_t$, $z^\psi_t = \displaystyle\land_{i=1}^m z^{\varphi_i}_{t_i}$, and $z^\psi_t = \displaystyle\lor_{i=1}^m z^{\varphi_i}_{t_i}$ for $\psi = \neg \mu$, $\psi = \displaystyle\land_{i=1}^m\varphi_i$ and $\psi = \displaystyle\lor_{i=1}^m\varphi_i$, respectively. These constraints enforce that $z^\psi_t = 1$ if $\psi$ holds at time $t$ and $z^\psi_t = 0$ otherwise.

\subsubsection{Temporal constraints}
We first present encodings for the $\G$ and $\F$
operators. We will use these encodings to define the encoding for the ${\cal U}_{[a,b]}$ operator.\\
%Given an interval $[a,b]$ defined relative to time $t$, let $a^N_t = \min(t+a,N)$, $\hat{a}^N_t = \max(0,\min(t+a-N,N))$,and $b^N_t = \min(t+b,N)$, $\hat{b}^N_t = \max(0,\min(t+b-N,N)$. We will use $a^N_t$, $\hat{a}^N_t$, $b^N_t$ and $\hat{b}^N_t$ to wrap the interval around the lasso-shaped trajectory parametrization.\\

\noindent \underline{Always: $\psi = \G_{[a,b]} \varphi$}\\

\noindent Let $a^N_t = \min(t+a,N)$ and $b^N_t = \min(t+b,N)$\\
Define
$
z^\psi_t =  \land_{i = a^N_t}^{b^N_t}z^\varphi_{i} \land (\bigvee_{j=1}^{N} l_j \land  \bigwedge_{i = \hat{a}^N_j}^{\hat{b}^N_j}z^\varphi_{i})
%z^\psi_t =  \land_{i = a^N_t}^{b^N_t}z^\varphi_{i}
%z^\psi_t =  \land_{i = a^N_t}^{b^N_t}z^\varphi_{i} \land (\lor_{j=1}^{N} l_j \land  \bigwedge_{i = {a}^{N}_{j+t}}^{{b}^{N}_{j+t}}z^\varphi_{i})
$

\noindent The logical operation $\land$ on the variables $z^\varphi_{i}$ here is as defined in Section
\ref{encode:bool}. Intuitively, this encoding enforces that the formula $\varphi$ is satisfied at every time step on the interval $[a,b]$ relative to time step $t$.\\

\noindent \underline{Eventually: $\psi = \F_{[a,b]} \varphi$}\\

%Let $a^N_t = \min(t+a,N)$ and $b^N_t = \min(t+b,N)$\\
\noindent Define
$
z^\psi_t =  \lor_{i = a^N_t}^{b^N_t}z^\varphi_{i} \land (\bigvee_{j=1}^{N} l_j \land  \bigvee_{i = \hat{a}^N_j}^{\hat{b}^N_j}z^\varphi_{i})
%z^\psi_t =  \lor_{i = a^N_t}^{b^N_t}z^\varphi_{i}
%z^\psi_t =  \lor_{i = a^N_t}^{b^N_t}z^\varphi_{i} \land (\lor_{j=1}^{N} l_j \land  \bigvee_{i = {a}^{N}_{j+t}}^{{b}^{N}_{j+t}}z^\varphi_{i})
$
%The logical operation $\land$ on the variables $z^\varphi_{i}$ here is as defined in Section \ref{encode:bool}.
This encoding enforces that the formula $\varphi$ is satisfied at some time step on the interval $[a,b]$ relative to time step $t$.\\

\noindent \underline{Until: $\psi = \varphi_1~{\cal U}_{[a,b]}~\varphi_2$}\\

\noindent The bounded until operator  ${\cal U}_{[a,b]}$ can be defined in terms of the unbounded ${\cal U}$ (inherited from LTL) as follows \cite{DonzeFM13}:
\[
\varphi_1~{\cal U}_{[a,b]}~\varphi_2 = \G_{[0,a]}\varphi_1 \land \F_{[a,b]}\varphi_2 \land \F_{[a,a]}(\varphi_1~{\cal U}~\varphi_2)
\]

We will use the encoding of the unbounded $\cal U$ from \cite{BiereHJLS06}. When encoding over infinite trajectories, this requires an auxiliary encoding that prevents the pitfalls of circular reasoning on the finite parametrization of the infinite sequences. The interested reader is referred to \cite{BiereHJLS06} for the details of the encoding. %for infinite trajectories.
The auxiliary encoding of the unbounded until is\\  $\llangle\varphi_1~{\cal U}~\varphi_2\rrangle_t =$\\
 $~~~~~~
   \begin{cases}
     z^{\varphi_2}_t \lor (z^{\varphi_1}_t \land \llangle\varphi_1~{\cal U}~\varphi_2\rrangle_{t+1}), & t = 1, ..., N-1\\
    z^{\varphi_2}_N. &
   \end{cases}
  $\\

\noindent With this definition in place, we define 
\[
z^{\varphi_1~{\cal U}~\varphi_2}_t = z^{\varphi_2}_t \lor (z^{\varphi_1}_t \land z^{\varphi_1~{\cal U}~\varphi_2}_{t+1})
\]
for $t = 1, ..., N-1$, and 
\[
z^{\varphi_1~{\cal U}~\varphi_2}_N = z^{\varphi_2}_N \lor (z^{\varphi_1}_N \land (\bigvee_{j=1}^{N} (l_j \land\llangle\varphi_1~{\cal U}~\varphi_2\rrangle_{j}))).
%z^{\varphi_1~{\cal U}~\varphi_2}_N.
\]

\noindent Given this encoding of the unbounded until and the encodings of $\G_{[a,b]}$ and $\F_{[a,b]}$ above, we can encode
\[
z^{\varphi_1~{\cal U}_{[a,b]}~\varphi_2}_t = z^{\G_{[0,a]}\varphi_1}_t \land z^{\F_{[a,b]}\varphi_2}_t \land z^{\F_{[a,a]}(\varphi_1~{\cal U}~\varphi_2)}_t.
\]
%Although this encoding of the until operator appears to require a number of variables quadratic in $N$, subformulas can be shared to use only a linear number of variables (i.e., by sharing the variables encoding $\varphi_1~{\cal U}~\varphi_2$, for example).

\noindent By induction on the structure of STL formulas $\varphi$, $z^\varphi_t = 1$ if and only if $\varphi$ holds on the system at time $t$. With this motivation, given a specification $\varphi$, we add a final constraint:
\begin{equation}\label{spec_holds}
z^\varphi_0 = 1.
\end{equation}

For a bounded horizon formula, the union of the STL constraints, loop constraints and system constraints gives the MILP encoding of
Problem \ref{prob:open}; this enables checking feasibility of this set of constraints and finding a solution
using an MILP solver. Given an objective function on runs of the system, this approach also enables finding the optimal open-loop trajectory that satisfies the STL specification. 
Algorithm \ref{alg:openloop} reviews the procedure for solving Problem \ref{prob:open}.

\begin{algorithm}
\begin{algorithmic}[1]
\Procedure{OPEN\_ LOOP}{$f, x_0, \w, N, \varphi, J$}
\State LOOP\_CONSTRAINTS $\leftarrow$ Sec. \ref{encode:lasso}
\State SYSTEM\_CONSTRAINTS  $\leftarrow$ Sec. \ref{encode:sys_dyn}
\State STL\_CONSTRAINTS $\leftarrow$ Sec. \ref{encode:bool} OR Sec. \ref{robust}
\State \label{step:1}
\[
\begin{array}{lll}
\displaystyle {\bf u}^* \leftarrow \argmin_{{\bf u}\in {\cal U}^N} && J(x_0, \u, \w, \f)\\
&\mbox{s.t. }&\text{LOOP\_CONSTRAINTS}\\
&&\text{SYSTEM\_CONSTRAINTS}\\
&&\text{STL\_CONSTRAINTS}
\end{array}
\]
Return ${\bf u}^*$
\EndProcedure
\caption{Algorithm for Problem \ref{prob:open}}
\label{alg:openloop}
\end{algorithmic}
\end{algorithm}

%%% Local Variables: 
%%% mode: latex
%%% TeX-master: "root_tac_part1"
%%% End: 

\subsection{Quantitative Encoding}\label{robust}
The robustness of satisfaction of the STL specification, as defined in \ref{robust_defn}, provides a natural objective for the MILP defined in Section \ref{sec:milp}, either in the absence of, or as a complement to domain-specific objectives on runs of the system. The robustness can be computed recursively on the structure of the formula in conjunction with the generation of constraints. Moreover, since $\max$ and $\min$ operations can be expressed in an MILP formulation using additional binary variables, this does not add complexity to the encoding, although the additional variables do make it more computationally expensive in practice. 

In this section, we sketch the MILP encoding of the predicates and Boolean operators using the quantitative semantics; the encoding of the temporal operators builds on these encodings, as in Section \ref{sec:milp}. Given a formula $\varphi$, we introduce a variable $r^\varphi_t$, and an associated set of MILP constraints such that $r^\varphi_t > 0$ if and only if $\varphi$ holds at position $t$. We recursively generate the MILP constraints, such that $r^\varphi_0$ determines whether a formula $\varphi$ holds in the initial state. Additionally, we enforce $r^\varphi_t = \r^{\varphi}(\x,t)$.
 
For each predicate $\mu \in P$, we now introduce variables $r^\mu_t$ for time indices $t = 0,1,...,N$, and set $r^\mu_t = \mu(x_t)$. To define $r^\psi_t$, where $\psi$ is a Boolean formula, we inductively assume that each operand $\varphi$ has a corresponding variable $r^\varphi_t = \r^{\varphi}(\x,t)$. Then the Boolean operations are defined as:

\underline{Negation: $r^\psi_t = \neg r^\phi_t$}~~~~~~~~~~~~~
$
r^\psi_t = - r^\phi_t
$
\smallskip

\underline{Conjunction: $r^\psi_t = \displaystyle\land_{i=1}^mr^{\varphi_i}_{t_i}$}
\begin{align}
&\sum_{i=1}^m p^{\varphi_i}_{t_i} = 1\label{line1}\\
&r^\psi_t \le r^{\varphi_i}_{t_i}, i = 1,...,m\label{line2}\\
&r^{\varphi_i}_{t_i} - (1-p^{\varphi_i}_{t_i})M \le r^\psi_t \le r^{\varphi_i}_{t_i} + M(1-p^{\varphi_i}_{t_i})\label{line3}
\end{align}

\noindent where we introduce new binary variables $p^{\varphi_i}_{t_i}$ for $i = 1,...,m$, and $M$ is a sufficiently large positive number.
Then equation (\ref{line1}) enforces that there is one and only one $j \in \{1,...,m\}$ such that $p^{\varphi_j}_{t_i} = 1$, equation (\ref{line2}) ensures that $r^\psi_t$ is smaller than all $r^{\varphi_i}_{t_i}$, and equation (\ref{line3}) enforces that $r^\psi_t = r^{\varphi_j}_{t_j}$ if and only if $p^{\varphi_j}_{t_j} = 1$. Together, these constraints enforce that $r^\psi_t = \min_i(r^{\varphi_i}_{t_i})$.

\smallskip

\underline{Disjunction: $\psi =  \displaystyle\lor_{i=1}^mr^{\varphi_i}_{t_i}$}
is encoded similarly to conjunction, replacing (\ref{line2}) with 
$
r^\psi_t \ge r^{\varphi_i}_{t_i}, i = 1,...,m.
$
Using a similar reasoning to that above, this enforces $r^\psi_t = \max_i(r^{\varphi_i}_{t_i})$.

The encoding for bounded temporal operators is defined as in Section \ref{sec:milp}; robustness for the unbounded until is defined using $\sup$ and $\inf$ instead of $\max$ and $\min$, but these are equivalent on our finite trajectory representation with discrete time. By induction on the structure of STL formulas $\varphi$, this construction yields $r^\varphi_t  > 0$ if and only if $\varphi$ is satisfied at time $t$. Therefore, we can replace the constraints over $z^\varphi_t$ in Section \ref{sec:milp} by these constraints that compute the value of $r^\varphi_t$, and instead of (\ref{spec_holds}), add the constraint
$
r^\varphi_0 > 0.
$

Since we consider only the discrete time semantics of \ac{STL} in this work, the Boolean encoding in Section \ref{sec:milp} could be achieved by converting each formula to \ac{LTL}, and using existing encodings such as that in \cite{WolffTM14}. However, the robustness-based encoding we presented in this section has no natural analog for \ac{LTL}. The advantage of this encoding is that it allows us to maximize the value of $r^\varphi_0$, obtaining a trajectory that maximizes robustness of satisfaction. Additionally, an encoding based on robustness has the advantage of allowing the STL constraints to be softened or hardened as necessary. For example, if the original problem is infeasible, we can allow $\r^\varphi_0 > -\epsilon$ for some $\epsilon > 0$, thereby easily modifying the problem to allow a limited violation of the STL property.

The disadvantage is that it is more expensive to compute, due the the additional binary variables
introduced during each Boolean operation. Additionally, including robustness as an objective makes
the cost function inherently non-convex, with potentially many local minima, and harder to
optimize. On the other hand, the robustness constraints are more easily relaxed, allowing us to
use a simpler cost function, which can make the problem more tractable.

\subsection{Complexity}
In general, our synthesis algorithm has the same complexity as \ac{MILP}s, which are NP-hard, hence computationally challenging when the dimensions of the problem grow.  It is nevertheless appropriate to characterize the computational costs of our encoding and approach in terms of the number of variables and constraints in the resulting MILP. In practice, one measure of problem size is the number of binary variables required to indicate the satisfaction of the predicates $\mu$. This depends directly on the number of predicates used in the STL formula $\varphi$. 

For the Boolean encoding, if $P$ is the set of predicates used in the formula, then $O(N \cdot |P|)$ binary variables are introduced. In addition, continuous variables are introduced during the MILP encoding of the STL formula. The number of continuous variables used is $O(N \cdot |\varphi|)$, where $|\varphi|$ is the length (i.e. the number of operators) of the formula. %Finally, loop constraints introduce $N$ additional binary variables.

For the robustness-based encoding, $O(N \cdot |P|)$ continuous variables are introduced (one per predicate per time step). In addition, binary variables are introduced during the MILP encoding of each
operator in the STL formula. The number of binary variables used is thus $O(N \cdot |\varphi|)$, where $|\varphi|$ is the number of operators of the formula. 

Our synthesis algorithm also has polynomial runtime for the following fragment of STL.

\begin{definition}[SNN-STL]
\emph{Safe Negation-Normal STL (SNN-STL)} is the fragment of STL generated by the recursive grammar
$$
\f ::= \pi^\mu \mid \neg \pi^\mu \mid \f_1 \land \f_2 \mid \G_{[a,b]}~\f
$$
\end{definition}

SNN-STL has the following properties:
\begin{itemize}
\item All negations appear only on atomic propositions (pushed down to the leaf nodes of the formula abstract syntax tree).
\item The only temporal operators are $\G$ (with unbounded and bounded intervals).
\item Only conjunctions are allowed, no disjunctions.
\end{itemize}

Such specifications are expressive enough to enforce, e.g., safety specifications in environments where the system state is confined to a conjunction of polyhedra. 

Let 
\[
\mathtt{OPEN\_LOOP\_NO\_STL}(f, x_0, N, \varphi, J)
\]

\noindent denote the procedure that is identical to Algorithm \ref{alg:openloop}, except that the optimization problem in Step \ref{step:1} is solved with $\mathtt{STL\_CONSTRAINTS} = \emptyset$.

\begin{theorem}[Polynomial-time Synthesis for SNN-STL]\label{thm:snn}
Suppose that $\varphi$ is in SNN-LTL and has linear predicates. Then if $\mathtt{OPEN\_LOOP\_NO\_STL}(f, x_0, \w, N, \varphi, J)$
is convex, so is $\mathtt{OPEN\_LOOP}(f, x_0, N, \varphi, J)$.
\end{theorem}
\begin{proof}
The proof proceeds by induction on the structure of the formula, showing that the constraints added by $\mathtt{STL\_CONSTRAINTS}$ for each operator
restrict the solution to a convex set. 
First note that since the predicates $\pi^\mu$ are linear, negation of predicates preserves convexity, since $\neg \pi^\mu$ is also linear.
Because the intersection of convex sets is convex, the conjunction of a set of convex constraints is also convex. Finally, since the $\G$ operator is implemented
in terms of conjunctions, the constraints imposed by $\G$ also preserve convexity of the resulting optimization problem.
%For the Boolean encoding, we see that the constraints introduced are of the form: 
%
%$
%\begin{array}{l}
%z^\psi_t - z^{\varphi_i}_{t_i} \le 0\\
%- z^\psi_t  + 1 - m + \sum_{i=1}^mz^{\varphi_i}_{t_i} \le 0
%\end{array}
%$
%
%These are all convex constraints.
%
%For the robustness-based encoding, the relevant constraints are:
%
%$
%1 - \sum_{i=1}^m p^{\varphi_i}_{t_i} \le 0\\
%\sum_{i=1}^m p^{\varphi_i}_{t_i} - 1 \le 0\\
%r^\psi_t - r^{\varphi_i}_{t_i} \le 0\\
%r^{\varphi_i}_{t_i} - (1-p^{\varphi_i}_{t_i})M - r^\psi_t \le 0,\\
%r^\psi_t - r^{\varphi_i}_{t_i} - M(1-p^{\varphi_i}_{t_i}) \le 0
%$
%
%which are again convex.
\end{proof}

Informally, Theorem \ref{thm:snn} states that our encoding of SNN-STL constraints into an MILP preserves convexity in the resulting optimization problem.
The resulting optimization problem is therefore encodable as an \ac{LP}, i.e. without the use of integer variables.

\begin{corollary}
Algorithm \ref{alg:openloop} is polynomial-time for SNN-STL specifications $\varphi$.
\end{corollary}

%%% Local Variables: 
%%% mode: latex
%%% TeX-master: "root_tac_part1"
%%% End: 

\section{Model Predictive Control Synthesis}\label{mpc_synth}
In this section, we will describe a solution to Problem \ref{prob:mpc} by adding \ac{STL} constraints to an \ac{MPC} problem formulation. At each step $t$ of the MPC computation, we will search for a finite trajectory of fixed horizon length $H$, such that the accumulated trajectory satisfies $\varphi$. 

\begin{figure}
\begin{center}
\begin{tikzpicture}[node distance=.8 cm, >=stealth',every join/.style=->,
box/.style={draw, text centered, text width=2cm, minimum height=1cm},
nobox/.style={text centered, text width=1.6cm,
},transform shape,scale=0.9
]

\node (sol) [nobox] {$\u*=$ $u^*_0u^*_1\hdots u^*_{N-1}$};
%\node (solver) [box, right=of sol] {MILP Solver};
%\node (milp) [box, right= of solver] {Parametric MILP};
\node (synthesis) [box, right=of sol] {Synthesis};
%\node (in) [nobox, above=2cm of milp] {$\f,\Sigma_d,J$};
\node (in) [nobox, above=2cm of synthesis] {$\f,\Sigma_d,J$};
%\node (controller) [above=.5cm of solver] {\bf Controller}; 
\node (controller) [above=.5cm of sol] {\bf Controller}; 

%\node (plant) [box, below= 2cm of solver, fill=black!5, text width=2cm] {\textbf{Plant $\Sigma$}\\[.2cm]  $x_0$};
\node (plant) [box, below= 2cm of synthesis, fill=black!5, text width=2cm] {\textbf{Plant $\Sigma$}\\[.2cm]  $x_0$};
\node (w0p) [nobox, below= .5cm of plant] {$w_0$};
%\node (w0d) [nobox, below= .5cm of milp] {$w_0w_1\hdots w_{N-1}$};
\node (w0d) [nobox, below= .5cm of synthesis] {$w_0w_1\hdots w_{N-1}$};

\begin{pgfonlayer}{background}
    %\node [draw,inner sep = 1.4cm, minimum height=1.8cm, fill=black!5,fit=(milp) (solver) (sol)] {};
    \node [draw,inner sep = 1.4cm, minimum height=1.8cm, fill=black!5,fit=(synthesis) (sol)] {};
  \end{pgfonlayer}

%\draw[->, thick] (in) -- node [pos=0.2,left] {synthesis} (milp);
\draw[->, thick] (in) -- node [pos=0.2,left] {} (synthesis);
\draw[->] (w0d) -- (synthesis);  
%\draw[->] (w0d) -- (milp);  
\draw[->] (w0p) -- (plant);  

{[start chain]
%\chainin (milp) [join];
\chainin (synthesis) [join];
%\chainin (solver) [join];
\chainin (sol) [join];
}

\draw[->] (sol) -- ++(-1.8,0) -- ++(0,-3.15) -- (plant) node [pos=0.85, above] {$u_0^*$};
%\draw[->] (plant) -- ++(5,0) node [pos=0.1, above] {$x_1$} -- ++(0,3.15) -- (milp) node [pos=0.6, above] {$x_0$};
\draw[->] (plant) -- ++(5,0) node [pos=0.1, above] {$x_1$} -- ++(0,3.15) -- (synthesis) node [pos=0.6, above] {$x_0$};

\end{tikzpicture}
 \end{center}

\caption{Closed-loop (MPC) problem formulation. As in the open loop scenario, a sequence of control inputs is synthesized from
  the specifications, dynamics and cost function. However, at each time step, only the first computed
input is used by the plant.}  
\label{fig:closed}
\end{figure}
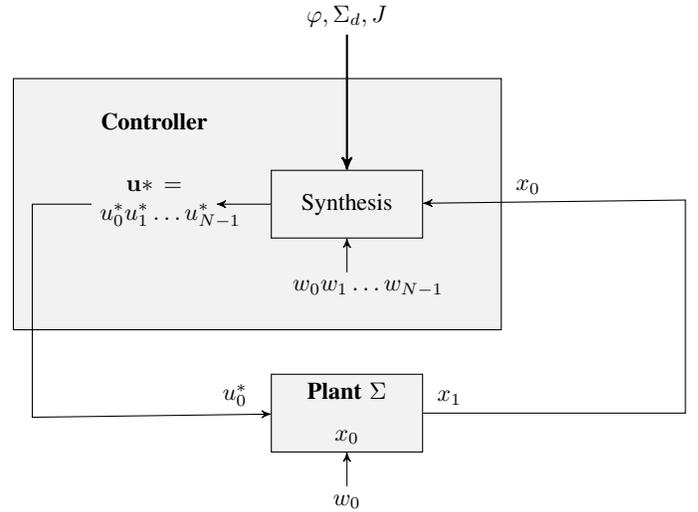

\subsection{Synthesis for bounded-time STL formulas}\label{mpc_bounded}
The length of the horizon $H$ is chosen to be at least the bound of formula $\varphi$. 
%We assume that we are generating control for a finite system trajectory of length $N$ using the MPC paradigm. 
At time step $0$, we will synthesize control ${\bf u}^{H,0}$ using the open-loop formulation in Section \ref{openloop}, including the STL constraints on the length-$H$ trajectory, but without the loop constraints. We will then execute only the first time step ${u}^{H,0}_0$. 
At the next step of the \ac{MPC}, we will solve for ${\bf u}^{H,1}$, while constraining the previous values of $x_0, u_0$ in the MILP, and the STL constraints on the trajectory up to time $H$. In this manner, we will keep track of the history of states in order to ensure that the formula is satisfied over the length-$H$ prefix of the trajectory, while solving for ${\bf u}^{H,t}$ at every time step $t$. %In this formulation, while the initial problem solves for the entire length of the trajectory, each subsequent MILP is further constrained based on the past, and so the problem at time $t+1$ is always smaller than that at time $t$.

\subsection{Extension to unbounded formulas}\label{mpc_unbounded}
For certain types of unbounded formulas, we can stitch together trajectories of length $H$ using a receding horizon approach, to produce an infinite computation that satisfies the STL formula. An example of this is safety properties, i.e. $\varphi = \G (\varphi_{MPC})$ for bounded STL formulas $\varphi_{MPC}$. For such formulas, at each step of the \ac{MPC} computation, we will search for a finite trajectory of horizon length $H$ (determined from $\varphi_{MPC}$ as above) that satisfies $\varphi_{MPC}$.

We now describe this approach in more detail. At each step $t$ of the receding horizon control computation, we will employ the open-loop approach in Section \ref{openloop} to find a finite trajectory of fixed horizon length $H$, such that the trajectory accumulated over time satisfies $\varphi$. Given a specification $\varphi=\G\varphi_{MPC}$, where $\varphi_{MPC}$ is a bounded-time formula with bound $H$. In this case, we can stitch together trajectories of length $H$ using a receding horizon approach to produce an infinite computation that satisfies the STL formula. At each step of the receding horizon computation, we search for a finite trajectory of horizon length $2H$, keeping track of the past values and robustness constraints necessary to determine satisfaction of $\varphi$ at every time step in the trajectory.  Note that we omit the loop constraints in this approach, because at each step we search for a finite trajectory, rather than an infinite trajectory with a finite parametrization.

First we define a procedure 
\[
\mathtt{OPEN\_LOOP}^*(f, x_0, \w, N, \G\varphi_{MPC}, J,{\bf P}^H,{\bf u}_{old}^t)
\]
that takes additional inputs ${\bf P} = \{P_0, P_1, ..., P_{H-1}\}$ and ${\bf u}_{old}^t = u_0,u_1,...,u_{t-1}$, and is identical to Algorithm \ref{alg:openloop}, except that the optimization problem posed in Step \ref{step:1} is solved without the loop constraints, and with the added constraints:
\[
\begin{array}{l}
\r^{\varphi}({\bf f}(x_0,{\bf u},{\bf w}), i) > P_i ~\forall i \in [0, H-1]\\
{\bf u}[i...t] = {\bf u}_{old}^t
\end{array}
\]

We then define a receding horizon control procedure as in Algorithm \ref{alg:mpc}. At each step, we are optimizing over a horizon of $2H$. We assume available a method $\mathtt{PREDICT\_W}(t)$ for predicting the sequence of $2H$ environment inputs starting at time step $t$.

\begin{algorithm}
\begin{algorithmic}[1]
\Procedure{MPC}{$f, x_0, \phi = \G \varphi_{MPC}, J$}
%\Require $f, x_0, \psi = \G \varphi, J$
%\Ensure ${\bf u}$ s.t. $\forall {\bf w}^\omega, {\bf f}(x_0, {\bf u},{\bf w})  \models \psi$
\State Let $M$ be a large positive constant. 
\State Let $H$ be the bound of $\varphi_{MPC}$.
%  - Transient steps: 
\State Set $P_0 = 0$ and $P_i = -M~\forall 0 < i \le H$. 
\State $\w^t \leftarrow \mathtt{PREDICT\_W}(0)$.
\State Compute ${\bf u}^0 = u_0^0,u_1^0,....,u_{2H-1}^0$ as:
\[
{\bf u}^0 \leftarrow \mathtt{OPEN\_LOOP}^*(f, x_0, \w^0, 2H, \G_{[0,H]}\varphi_{MPC}, J,{\bf P}^H,\emptyset)
\]
\For  {t=1; t<=H;t=t+1}
\State Set ${\bf u}_{old}^t = u_0^0,u_1^1,u_2^2,...,u_{t-1}^{t-1}$.
\State Set $P_i = 0$ for $0 \le i \le t$, $P_i = -M~\forall t < i \le H$.
\State $\w^t \leftarrow \mathtt{PREDICT\_W}(t)$.
\State Compute ${\bf u}^t = u_0^t,u_1^t,....,u_{2H-1}^t$ as:
\[
{\bf u}^t \leftarrow \mathtt{OPEN\_LOOP}^*(f, x_t, \w^t, 2H, \G_{[0,H]}\varphi_{MPC}, J,{\bf P}^H,{\bf u}_{old}^t)
\]
\EndFor
%  - Stationary steps: 
%\State Set ${\bf u}_{old}^t = u_0^0u_1^1u_2^2...u_{H-1}^{H-1}$.
%\State Set $P_i = 0$ for $0 \le i \le H$.
%\State Compute ${\bf u}^H = u_0^H,u_1^H,....,u_{2H-1}^H$ as:
%\[
%{\bf u}^H \leftarrow \mathtt{CEGIS}^*(f, x_H, 2H, \varphi, J,{\bf P}^H,{\bf u}_{old}^H)
%\]
%\State $t = t+1$.
\While  {$\mathtt{True}$}
\State Set ${\bf u}_{old}^t = u_1^{t-1},u_2^{t-1},u_3^{t-1},...,u_{t}^{t-1}$.
\State Set $P_i = 0$ for $0 \le i \le H$.
\State $\w^t \leftarrow \mathtt{PREDICT\_W}(t)$.
\[
{\bf u}^t \leftarrow \mathtt{OPEN\_LOOP}^*(f, x_t, \w^t, 2H, \G_{[0,H]}\varphi_{MPC}, J,{\bf P}^H,{\bf u}_{old}^t)
\]
\EndWhile
\EndProcedure
\caption{MPC Algorithm for Problem \ref{prob:mpc}}
\label{alg:mpc}
\end{algorithmic}
\end{algorithm}

Algorithm \ref{alg:mpc} has two phases, a \emph{transient} phase (Lines 4-10) and a \emph{stationary} phase (Lines 11-14). The transient phase applies until an initial control sequence of length $H$ has been computed, and the stationary phase follows. In the transient phase, the number of stored previous inputs (${\bf u}^t_{old}$) as well as the number of time steps at which formula $\varphi_{MPC}$ is enforced (i.e. time steps for which $P_i = 0$) grows by one at each iteration, until they both attain a maximum of $H$ at iteration $H$. Every following iteration uses a window of size $H$ for stored previous inputs, and sets all $P_i = 0$. The size-$H$ window of previously-computed inputs advances forward one step in time at each iteration after step $H$. In this manner, we keep a record of the previously computed inputs required to ensure satisfaction of $\varphi_{MPC}$ up to $H$ time steps in the past.

We now show that if Algorithm \ref{alg:mpc} does not terminate, then the resulting infinite sequence of control inputs enforces satisfaction of the specification $\phi = \G \varphi_{MPC}$.

\begin{theorem}
Let $\phi = \G \varphi_{MPC}$, and assume that ${\bf u}^*$ is an infinite sequence of control inputs generated by setting ${\bf u}^*[t] = u^t_0$, where ${\bf u}^t=u^t_0u^t_1...u^t_{2H-1}$ is the control input sequence of length $2H$ generated by Algorithm \ref{alg:mpc} at time $t$. Then %$\forall {\bf w}\in W^\omega,~
${\bf f}(x_0, {\bf u}^*,{\bf w}) \models \varphi$.
\end{theorem}
\begin{proof}
%We want to show that $\varphi$ is satisfied at every time $t$, $\forall {\bf w}\in W^\omega$. 
%\emph{Transient Phase $t \le H$:} 
%By the correctness of Algorithm \ref{alg:cegis}, each ${\bf u}^t$ has the property that $\forall {\bf w}^{2H} \in W^{2H},~({\bf f}(x_t,{\bf u}^t,{\bf w}^{2H}),i) \models \varphi$ for $i \in [t,t+H]$. In the transient phase, the added constraint $\r^{\varphi}({\bf f}(x_t,{\bf u},{\bf w}^0), i) > 0$ for $i \le t$ ensures that $\varphi$ is satisfied over all time steps in $[0,t]$. The result at time step $t$ is that 
%\emph{Stationary Phase ($t > H$):} 
Since $H$ is the bound of $\varphi_{MPC}$, the satisfaction of $\varphi_{MPC}$ at time $t$ is established by the control inputs ${\bf u}^*[t:t+H-1]$. At time $t+H$, 
\[
\begin{array}{lll}
{\bf u}_{old}^{t+H} &=& u^{t+H}_0,u^{t+H}_1,u^{t+H}_2,...,u^{t+H}_{t+H-1} \\
&=& u^{t+H-1}_1,u^{t+H-1}_2,u^{t+H-1}_3,...,u^{t+H-1}_{t+H} \\
&=& u^{t}_{t},u^{t+1}_{t+1},u^{t+2}_{t+2},...,u^{t+H-1}_{t+H-1}\\
&=&{\bf u}^*[t:t+H-1],
\end{array}
\]
and so all the inputs required to determine satisfaction of $\varphi$ at time $t$ have been fixed. Moreover, if ${\bf u}^{t+H}$ is successfully computed, then by the correctness of Algorithm \ref{alg:openloop}, ${\bf u}_{old}^{t+H}$ has the property that %$\forall {\bf w}^{H} \in W^{H},~
${\bf f}(x_t,{\bf u}_{old}^{t+H},{\bf w}^{H}) \models \varphi_{MPC}$. Since ${\bf u}^*[t:t+H-1] = {\bf u}_{old}^{t+H}$, we see that %$\forall {\bf w}^{H} \in W^{H},~
${\bf f}(x_t,{\bf u}^*[t:t+h],{\bf w}^{H}) \models \varphi_{MPC}$. 

It follows that %$\forall {\bf w}^{\omega} \in W^{\omega},~
${\bf f}(x_0,{\bf u}^*,{\bf w}) \models \varphi_{MPC}$. 
\end{proof}

We have therefore shown how a control input can be synthesized for infinite sequences satisfying $\varphi$, by repeatedly synthesizing control for sequences of length $2H$. A similar approach applies for formulas $\F \varphi_{MPC}$ and $\varphi_{MPC}~{\mathcal U}~\psi_{MPC}$, where $\varphi_{MPC},\psi_{MPC}$ are bounded-time. 

Note that we assumed that $\mathtt{PREDICT\_W}(t)$ returns an exact prediction of the disturbance signal over the next $2H$ time steps. The correctness of our approach relies on this assumption. An interesting direction of future work is to relax this requirement, demanding only an uncertain prediction of the disturbance signal.

\begin{remark}
The control objective for \ac{MPC} is usually to steer the state to the origin or to an equilibrium state. Questions that arise include those of ensuring feasibility at each time step, closed-loop stability and near-optimal performance \cite{MorariLee99}.  There is a mature theory of stability for \ac{MPC}, where the essential ingredients are terminal costs, terminal constraint sets, and local stabilizing controller that ensure closed-loop stability \cite{MayneRRS00}.

In this work, our control objective is not closed-loop stability, but satisfaction of an \ac{STL} formula. We achieve this, as detailed above, through choice of a sufficiently large prediction horizon $H$. This can be compared with the manner in which automatic satisfaction of a terminal constraint is sometimes attained by prior choice of a sufficiently large horizon.
\end{remark}

%\begin{algorithm}
%\begin{algorithmic}[1]
%\Procedure{MPC}{$\{\varphi, x_0, J\}$}
%\State Compute the system constraints for a trajectory of length $2N$.
%\State Introduce the constraints $\r^{\varphi}({\bf f}(x_0,{\bf u},{\bf w}), t) > P_i$ for all $t \in [0, N-1]$.
%\State Let $M$ be a large positive constant.
%%  - Transient steps: 
%\State Compute optimal $u_0^0,u_1^0,....,u_{2N-1}^0$ while constraining $P_0 = 0$ and $P_i = -M~\forall i > 0$.
%\For  {t=1; t<N;t++}
%\State Set $u_0^t = u_0^0, u_1^t = u_1^1, ..., u_{t-1}^t = u_{t-1}^{t-1}$.
%\State Constrain $P_i = 0$ for $0 \le i \le t$, and $P_i = -M~\forall i > t$.
%\State Compute optimal $u_0^t,u_1^t,....,u_{2N-1}^t$.
%\EndFor
%%  - Stationary steps: 
%\State Set $u_0^N = u_0^0, u_1^N = u_1^1, ..., u_{N-1}^N = u_{N-1}^{N-1}$.
%\State Constrain $P_i = 0$ for $0 \le i \le N$.
%\State Compute optimal $u_0^N,u_1^N,....,u_{2N-1}^N$.
%\State $t = t+1$.
%\While  {$\mathtt{True}$}
%\State Set $u_0^t = u_1^{t-1}, u_1^t = u_2^{t-1}, ..., u_{t-2}^t = u_{t-1}^{t-1}$.
%\State Constrain $P_i = 0 \forall 0 \le i \le N$.
%\State Compute optimal $u_0^t,u_1^t,....,u_{2N-1}^t$.
%\EndWhile
%\EndProcedure
%\caption{Algorithm for Problem \ref{prob:stl_MPC}}
%\label{alg:MPC}
%\end{algorithmic}
%\end{algorithm}

%%% Local Variables: 
%%% mode: latex
%%% TeX-master: "root_tac_part1"
%%% End: 

\section{Experimental Comparison of Encodings}
%\subsection{Boolean vs Robust Encoding}
We implemented the Boolean and robust encodings using the tools Breach \cite{breach} and YALMIP
\cite{YALMIP}, and now present results obtained with the following formulas:
\begin{itemize}
\item $\f_1=\G_{[0,0.1]} x^{(1)}_t>0.1$
\item $\f_2=\G_{[0,0.1]} (x^{(1)}_t>0.1) \wedge \G_{[0,0.1]} (x^{(2)}_t<-0.5)$
\item $\f_3=\G_{[0,0.5]}\F_{[0,0.1]}(x_t^{(1)}>0.1)$
\item $\f_4=\F_{[0,0.2]} ( x^{(1)}_t>0.1 \wedge (\F_{[0,0.1]}(x_t^{(2)}>0.1)$\\ 
$~~~~~~~~~\wedge \F_{[0,0.1]} (x^{(3)}_t>0.1)))$
\end{itemize}
In this study, we used the trivial system $\x = \u$, where $\x$ is a 3-dimensional signal (i.e. $x_t = x_t^{(1)}x_t^{(2)}x_t^{(3)}$), so that
no constraint is generated for the system dynamics, and the cost function $J(\x,\u) =
\sum_{k=1}^{N} \|\u_{t_k}\|_1$. Note that the output of this procedure for a formula $\f$ is a signal
of minimal norm which satisfies $\f$ when using the Boolean encoding and which satisfies $\f$ with a
specified robustness $\r^\f(\x)=0.1$ for the robust encoding. For each formula we computed the
Boolean and robust encodings for an horizon $N=30$ and sampling time $\tau=0.025s$ and report the
number of constraints generated by each encoding, the time to create the resulting MILP with YALMIP
and the time to solve it using the solver Gurobi.\footnote{http://www.gurobi.com/} All experiments
were run on a laptop with an Intel Core i7 2.3 GHz processor and 16 GB of memory.

\begin{table}[ht]
\small
\centering
\begin{tabular}{@{\extracolsep{\fill}}ccccccc}
Formula  & \multicolumn{2}{c}{\#constraints} & \multicolumn{2}{c}{YALMIP Time (s)} & \multicolumn{2}{c}{Solver time (s)} \\
& B & R & B & R & B & R  \\
\hline
$\f_1$& 154 & 488  & 1.71 & 2.04 & 0.0070  & 0.0085 \\
$\f_2$& 364 & 897  & 1.94 & 2.69 & 0.0115 & 0.0229 \\
$\f_3$& 244 & 1282 & 1.84 & 3.15 & 0.0064 & 0.1356 \\
$\f_4$& 574 & 1330 & 2.29 & 3.37 & 0.2167 & 238.6\\
%\bottomrule
\end{tabular}
\caption{Boolean (B) vs robust (R) encodings. YALMIP time represents the time taken by the tool YALMIP in
  order to generate the MILP and Solver time is the time taken by the solver Gurobi to actually solve it.}
\label{tab:encodigs}
\end{table}

A first observation is that for both encodings, most of the time is spent creating the MILP, while
solving it is done in a fraction of a second. Also, while the robust encoding generates 3 to 5 times
more constraints, the computational time to create and solve the corresponding MILPs is hardly twice
more. The exception is solving the MILP for $\f_4$, which takes significantly more time for the
robust encoding than for the Boolean encoding. The reason is hard to pinpoint without a more
thorough investigation, but we can already note that solving a MILP is NP-hard, and while solvers
use sophisticated heuristics to mitigate this complexity, instances for which these heuristics
fail are bound to appear.

%%% Local Variables: 
%%% mode: latex
%%% TeX-master: "root_tac_part1"
%%% End: 

\section{Case Study: Building Climate Control}

\subsection{Mathematical Model of a Building} \label{sec:hvacModel}
Next we consider the problem of controlling building indoor climate,
using the model proposed by Maasoumy et al~\cite{Maasoumy13}. In this section we present a summary of the building's thermal model. % We then present an MPC formulation for the control of building HVAC system to minimize the energy consumption of a building while satisfying the state and input constraints. 
%\subsubsection{Heat Transfer}

As shown in Fig.~\ref{fig:perspective}, the building is modeled as a resistor-capacitor circuit with $n$ nodes, $m$ of which are rooms and the remaining $n-m$ are walls. We denote the temperature of room $r_i$ by $T_{r_i}$. The wall and temperature of the wall between rooms $i$ and $j$ are denoted by $w_{i,j}$ and $T_{w_{i,j}}$, respectively. The temperature of wall $w_{i,j}$ and room $r_i$ are governed by the following equations:

\begin{align}\label{Eq:WallRoom}
  C^{w}_{i,j} \frac{dT_{w_{i,j}}}{dt}= & \sum_{k \in \mathcal{N}_{w_{i,j}}} {\frac{T_{r_{k}} -T_{w_{i,j}}}{R_{i,j_{k}}}} + r_{i,j} \alpha_{i,j} A_{w_{i,j}} Q_{rad_{i,j}}\\
 C^{r}_{i} \frac{dT_{r_i}}{dt} = & \sum_{k \in \mathcal{N}_{r_i}}
{\frac{T_k-T_{r_i}}{R_{i,k_{i}}}} + \dot{m}_{r_i} c_a (T_{s_i}-T_{r_i}) +  \notag \\
& w_i \tau_{w_i} A_{win_i} Q_{rad_i} + \dot{Q}_{int_i},
\end{align}

\noindent where $C^{w}_{i,j}$, $\alpha_{i,j}$ and $A_{w_{i,j}}$ are
heat capacity, a radiative heat absorption coefficient, and the area
of $w_{i,j}$, respectively. $R_{i,j_{k}}$ is the total thermal
resistance between the centerline of wall $(i,j)$ and the side of the
wall on which node $k$ is located. $Q_{rad_{i,j}}$ is the radiative
heat flux density on $w_{i,j}$. $\mathcal{N}_{w_{i,j}}$ is the set of
all neighboring nodes to $w_{i,j}$. $r_{i,j}$ is a wall identifier, which
equals $0$ for internal walls and $1$ for peripheral walls, where
either $i$ or $j$ is the outside node. $T_{r_i}$, $C^{r}_{i}$ and
$\dot{m}_{r_i}$ are the temperature, heat capacity and air mass flow
into room $i$, respectively. $c_a$ is the specific heat capacity of
air, and $T_{s_i}$ is the temperature of the supply air to room
$i$. $w_i$ is a window identifier, which equals $0$ if none of the
walls surrounding room $i$ have windows, and $1$ if at least one of
them does. $\tau_{w_i}$ is the transmissivity of the glass of window
$i$, $A_{win_i}$ is the total area of the windows on walls surrounding
room $i$, $Q_{rad_i}$ is the radiative heat flux density per unit area
radiated to room $i$, and $\dot{Q}_{int_i}$ is the internal heat
generation in room $i$. $\mathcal{N}_{r_i}$ is the set of
neighboring \emph{room} nodes for room $i$.  Further details on this
thermal model can be found in~\cite{Maasoumy13}. 

\begin{figure}[t!]
\centering
\includegraphics[width=1.0\columnwidth]{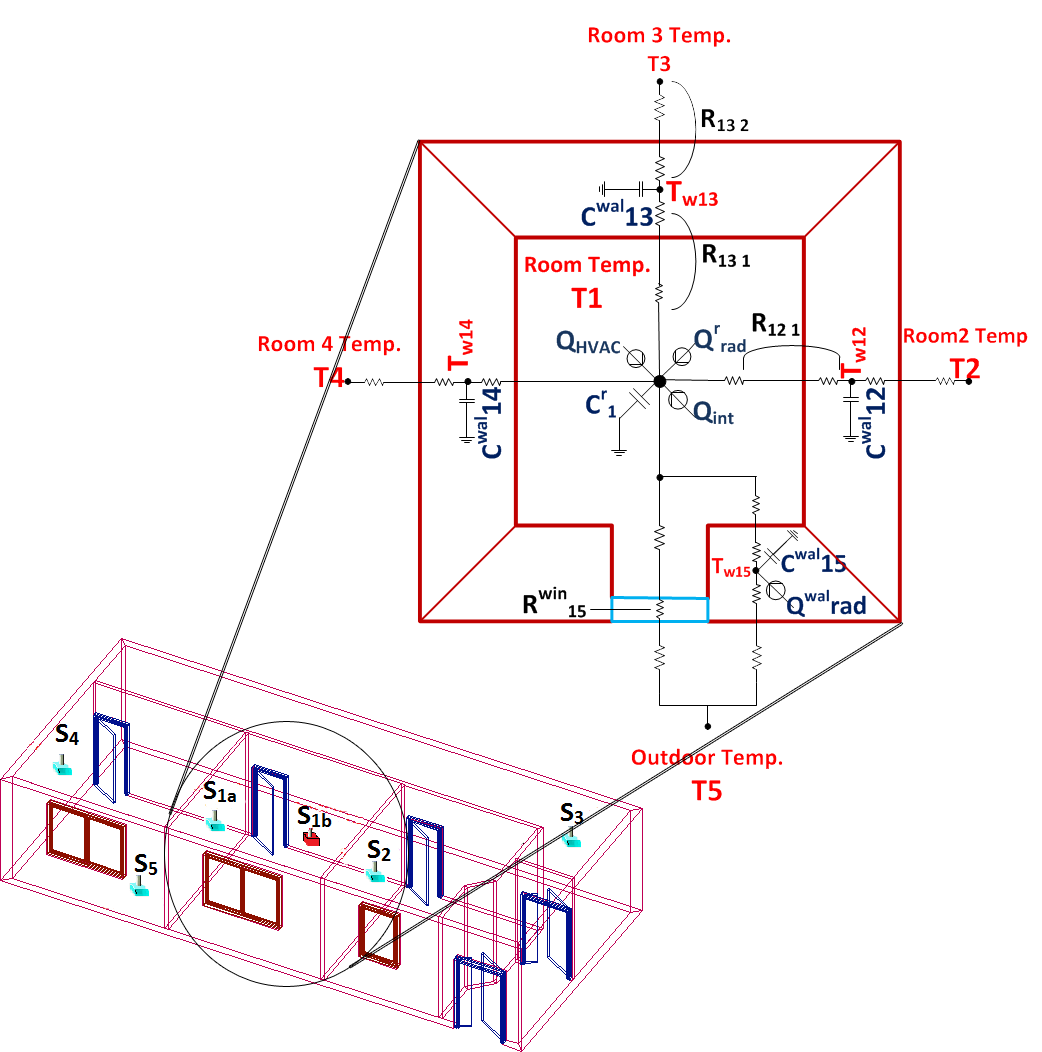} % .eps
\caption{Resistor-capacitor representation of a typical room with a window.}
\label{fig:perspective}
\end{figure}

The heat transfer equations for each wall and room yield the system dynamics:
\[
\dot{x}_{t}= f(x_{t},u_{t}, w_{t}).
%~~~y_{t} = C x_{t}
\]
\noindent Here $x_{t} \in \mathbb{R}^{n}$ is the state vector
representing the temperature of the nodes in the thermal network, and
$u_{t} \in \mathbb{R}^{lm}$ is the input vector representing the air
mass flow rate and discharge air temperature of conditioned air into
each thermal zone (with $l$ being the number of inputs to each thermal
zone, e.g. two for air mass flow and supply air temperature). 
The HVAC system of the building considered for this study operates
with a constant supply air temperature, while air mass flow is the time varying control input.
Hence, in the following simulations we consider supply air temperature constant and treat air mass flow as the control signal. Vector $w_t$
stores the estimated disturbance values, aggregating various unmodelled
dynamics such as $T_{out}$, $\dot{Q}_{int}$ and $Q_{rad}$, and can be
estimated using historical data~\cite{MaasoumyTotal}. $y_{t} \in
\mathbb{R}^{m}$ is the output vector, representing the temperature of
the thermal zones.
%, and $C$ is a constant matrix of proper dimension. 
The building model was trained using historical data, and the result of the system
identification is shown in Fig.~\ref{fig:ident_param_dist}.

\begin{figure} [!h]
\centering
\includegraphics[width=\columnwidth]{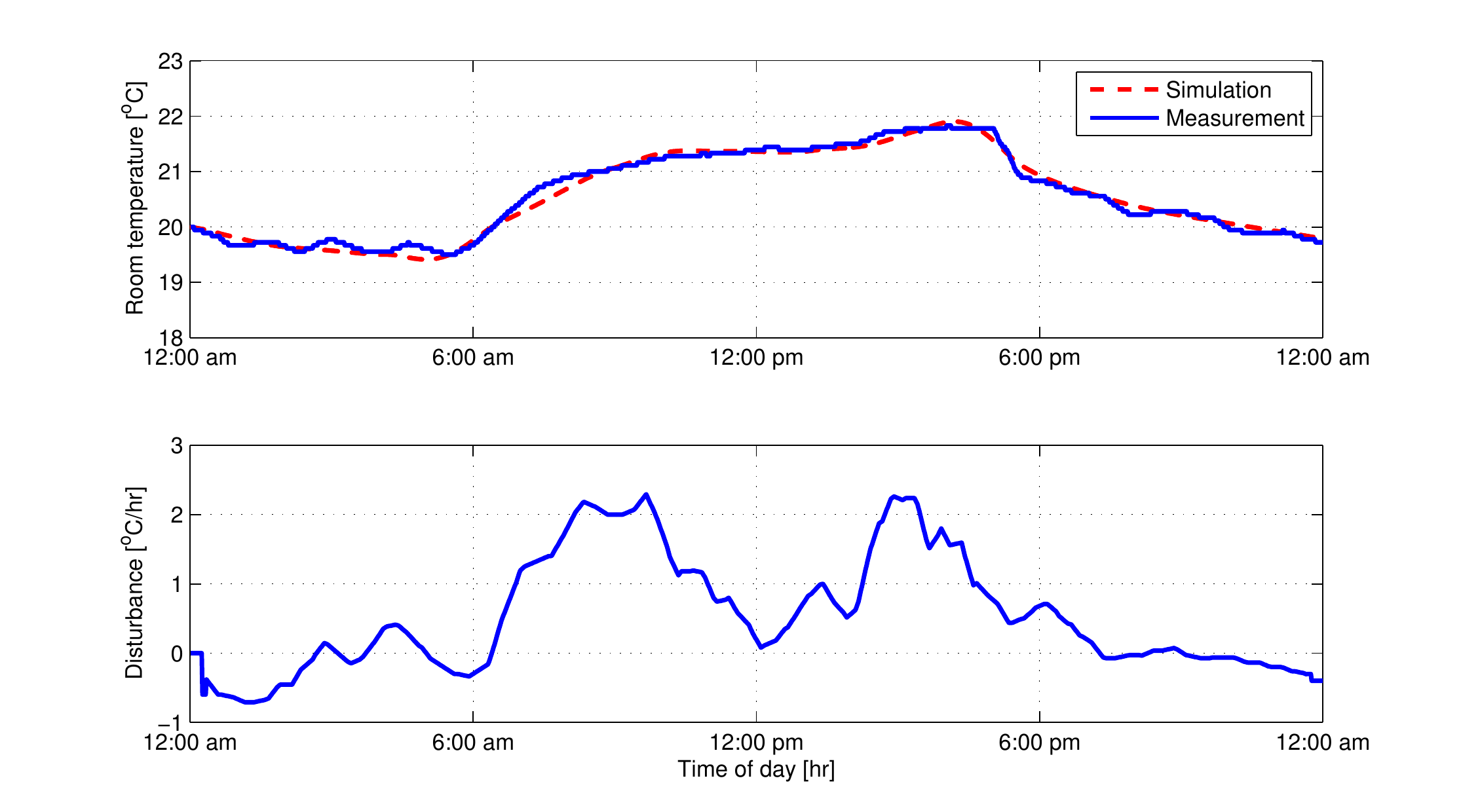}
\caption{Simulated temperature, measured temperature and unmodelled dynamics of a thermal zone in Bancroft library on UC Berkeley campus.}
\label{fig:ident_param_dist}
\end{figure} 

\subsection{MPC for Building Climate Control}
We consider the problem of controlling the above building's HVAC system using an MPC scheme. 
We adopt the MPC formulation proposed  by Maasoumy et al.~\cite{ACC:Selection}, with the objective of minimizing the total energy cost (in dollar value). $\tau$ and $H$ denote the length of each time slot and the prediction horizon (in number of time slots) of the MPC, respectively. Assume that the system dynamics are also discretized with a sampling time of $\tau$. Here we consider $\tau=0.5$ hr and $H=24$. At each time $t$, the predictive controller solves an optimal control problem to 
compute $\vec{u}_t=[u_t,\ldots,u_{t+H-1}]$, and minimizes the
cumulative norm of $u_t$: $\sum_{k=0}^{H-1} \|u_{t+k}\|$. We assume
known an occupancy function $\occ_t$ which is equal to 1 when the room
is occupied and to 0 otherwise. The purpose of the MPC is to maintain
a comfort temperature given by $T^{\comf}$ whenever the room is
occupied while minimizing the cost of heating. This problem can
be expressed as follows: 
\begin{align*}
 &\underset{\vec{u}_t}{\min}  \sum_{k=0}^{H-1} \|u_{t+k}\| ~~ \text{s.t.}\\
 &       x_{t+k+1}=f(x_{t+k},u_{t+k},w_{t+k}), \\
 &x_t \models  \f ~ \text{ with }  ~\f = \G_{[0, H]}  (( \occ_t > 0 ) \Rightarrow  (  T_t > T^{\comf}_t)  \\
 & u_{t+k}\in\mathcal{U}_{t+k},\ k=0,...,H-1
\end{align*}
\noindent 

The STL formula was encoded using the robust MILP encoding and results are presented in
Fig.~\ref{fig:hvac}. Again we observed that creating the MILP structure was longer than solving an
instance of it (4.1s versus 0.15s). However, by using a proper parametrization of the problem in
YALMIP, the creation of the MILP structure can be done once offline and reused online for each step of the
MPC, which makes the approach promising and potentially applicable even for real-time applications.

\begin{figure}
  \centering
  \includegraphics[trim = 10mm 0mm 5mm 5mm, clip, width=.5\textwidth]{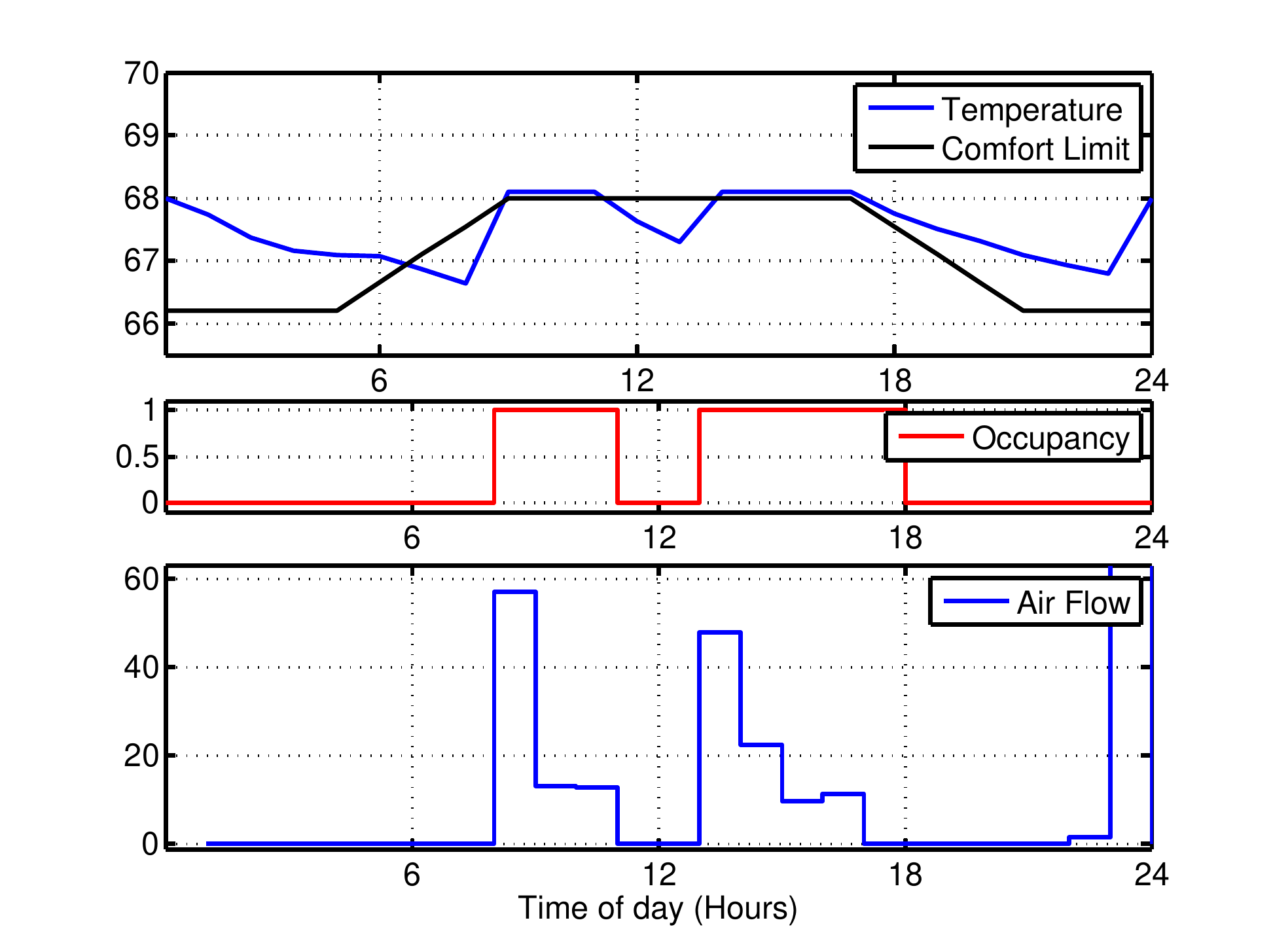}
  \caption{Room temperature control with constraints based on occupancy, expressed in STL.}
  \label{fig:hvac}
\end{figure}

%%% Local Variables: 
%%% mode: latex
%%% TeX-master: "root_tac_part1"
%%% End: 

\section{Case Study II: Regulation Control for Smart Grid}
\subsection{Mathematical Model}

\begin{figure}[t] %htbp
\centering
\includegraphics[width=0.9\columnwidth]{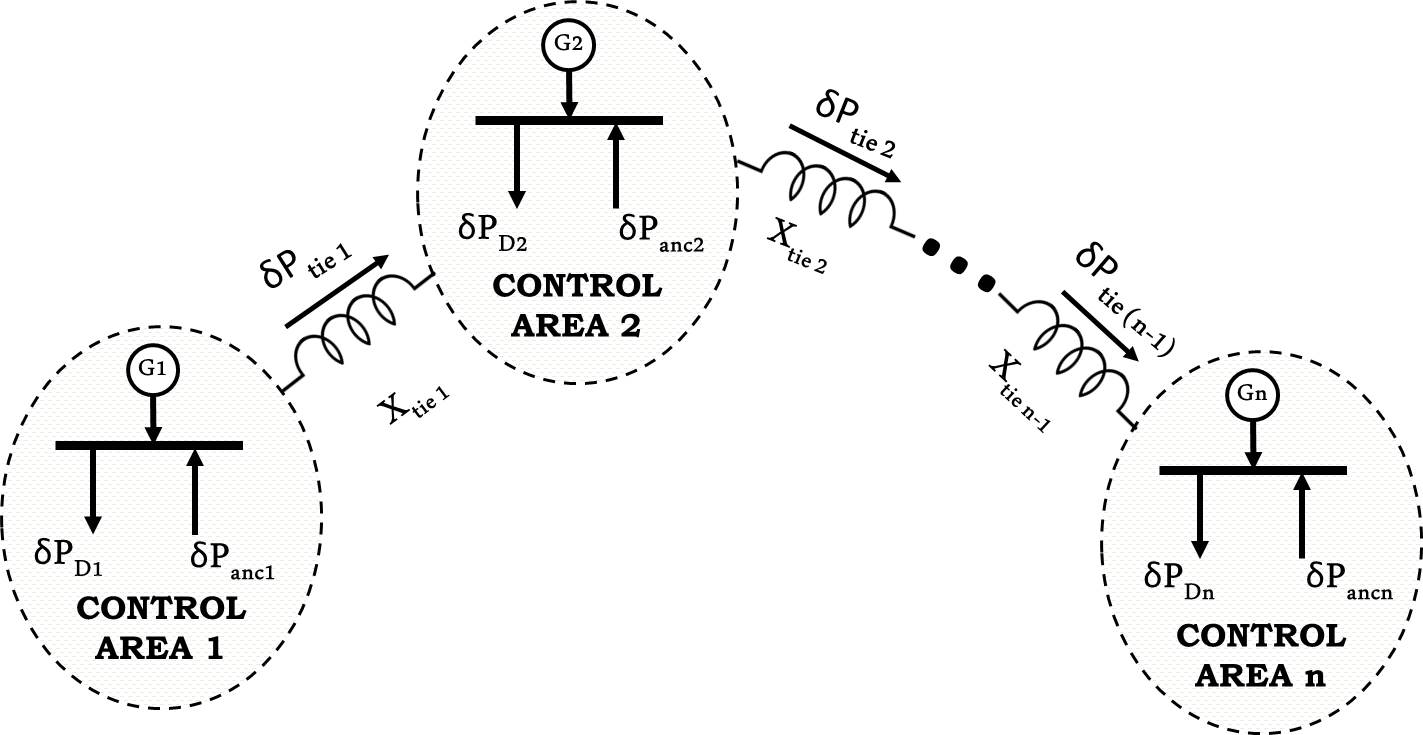}
\caption{Power system grid with $n$ areas. The dynamics in each area is depicted in Fig.\ref{fig:BlockDiag}.}
\label{fig:nAreas}
\end{figure}

\begin{figure}[t] %htbp
\centering
\includegraphics[width=1.0\columnwidth]{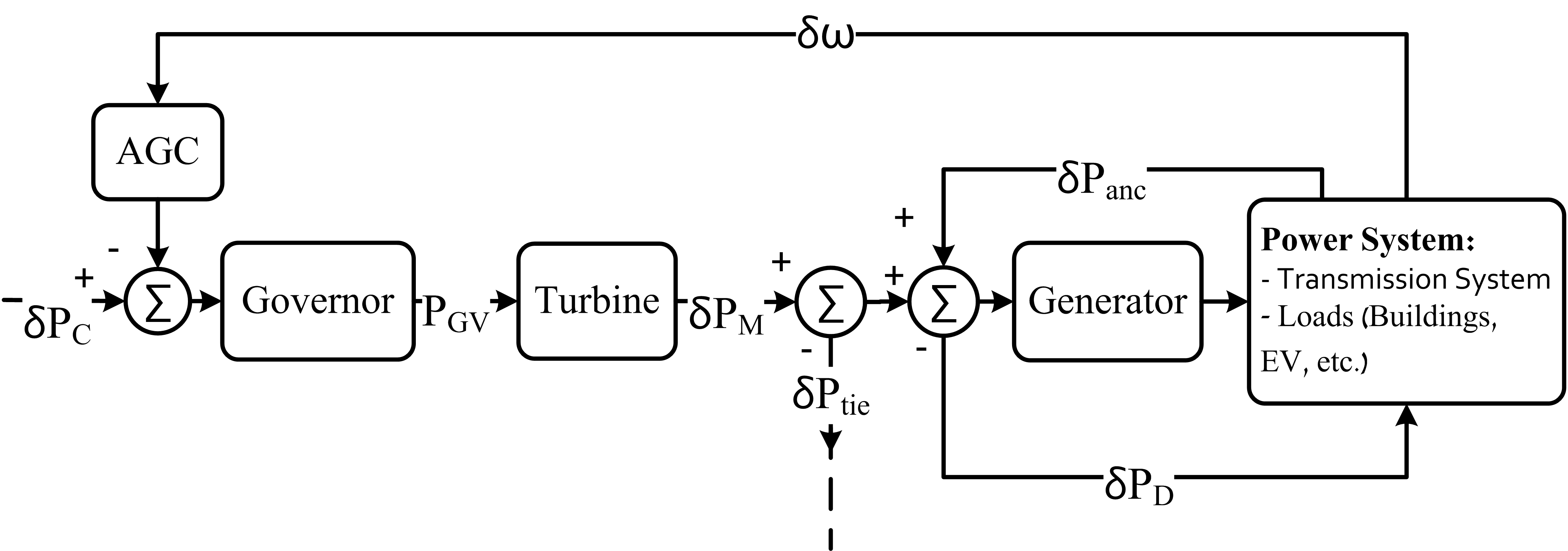}
\caption{Block diagram of power system and its relation to governor, turbine, generator, and the
  \ac{AGC} signal for each control area. More details on the power grid model can be found in~\cite{MaasoumySSP14}.}
\label{fig:BlockDiag}
\end{figure}

The second case study we consider is the $n$-areas smart grid model presented
in~\cite{MaasoumySSP14} and depicted in Fig.~\ref{fig:nAreas}. The interconnection of power system
components, including a governor, turbine and generator in each area is shown in the block diagram
in Fig.~\ref{fig:BlockDiag}. In the diagram, $\delta P_C$ is a control input which acts against an
increase or decrease in power demand to regulate the system frequency $\omega$, and $\delta P_D$
denotes fluctuations in power demand, modeled as an exogenous input (disturbance). Under steady
state, we have: $\omega = \omega_o$ and $P_M=P_G=P_M^o$, where $\omega_o$, $V_t^o$,
and $P_M^o$ are the nominal values for rated frequency, terminal voltage and mechanical power
input.

Next, we present the mathematical model for one area $i$ (note that superscripts refer to the
control area, and subscripts index states in each area).

\begin{subequations}
\label{Eq:PowSysDyn}
\scriptsize
\begin{align}
\frac{dx^{i}_1}{dt} &= \frac{(-D^{i} x^{i}_1 + \delta P^{i}_{M} - \delta P^{i}_{D} - \delta P^{i}_{\mathrm{tie}} + \delta P^{i}_{\anc})}{M} , \label{eq:x1} \\
\frac{dx^{i}_2}{dt} &= \frac{(x^{i}_3 - x^{i}_2)}{T^{i}_{7}},
~~~\frac{dx^{i}_3}{dt} = \frac{(x^{i}_4 - x^{i}_3)}{T^{i}_{6}},
~~~\frac{dx^{i}_4}{dt} = \frac{(x^{i}_5 - x^{i}_4)}{T^{i}_{5}}, \\
\frac{dx^{i}_5}{dt} &= \frac{(P^{i}_{GV} - x^{i}_5)}{T^{i}_{4}},
~~~\frac{dx^{i}_6}{dt} = \frac{(x^{i}_7 - x^{i}_6)}{T^{i}_{3}}, \\
\frac{dx^{i}_7}{dt} &= \frac{(-x^{i}_{7} + \delta P^{i}_{C} - x^{i}_1/R^{i})}{T^{i}_{1}},
~~~\frac{dx^{i}_{8}}{dt} = x^{i}_{1}
\end{align}
\end{subequations}
where $\delta P^{i}_{M}$ and $P^{i}_{GV}$ are given by
$\delta P^{i}_{M} = K^{i}_{1} x^{i}_{5} + K^{i}_{3} x^{i}_{4} +K^{i}_{5} x^{i}_{3} +K^{i}_{7}
x^{i}_{2}$,
and $P^{i}_{GV} = (1 - T_2 / T_3)x^{i}_{6} + (T_2 / T_3) x^{i}_{7}$. $D$ is the damping coefficient,
$M$ is the machine inertia constant, $R$ is the speed regulation constant, $T_i$'s are time
constants for power system components, and $K_i$'s are fractions of total mechanical power outputs
associated with different operating parts of the turbine. $\delta P^{i}_{\mathrm{tie}}$ represents
power transfer from area $i$ to other areas.  In equation~\eqref{Eq:PowSysDyn}, the first state represents
the frequency increment, $x^{i}_1 = \delta \omega_{i}$. It can be shown that $P^{i}_{\mathrm{tie}}$
can be obtained from
\begin{equation}
\delta P^{ij}_{\mathrm{tie}} = \sum^{n}_{j=1} \nu_{ij} (x^{i}_8 - x^{j}_8),
\end{equation}
where $\nu_{ij}$ is the transmission line stiffness coefficient, and the state variable $x^{i}_{8}$
is the integral of $x_1^i$.\\

The classical automatic generation control (\ac{AGC}) implements a simple PI control to regulate the grid frequency.  
In a multi-area power system, in addition to regulating frequency within each area, the auxiliary control should maintain the net interchange power with neighboring areas at scheduled values \cite{Bevrani09}. This is generally accomplished by adding a tie-line flow deviation to the frequency deviation in the auxiliary feedback control loop. A suitable linear combination of the frequency and tie-line deviations for area $i$, is known as the \ac{ACE}: this measures the difference between the scheduled and actual electrical generation within a control area while taking frequency bias into account. 
The \ac{ACE} of area $i$ is thus defined as $ACE^i = \delta P^{i}_{tie} + \beta^{i} x^{i}_{1}$, and
$\beta^{i}$ is the bias coefficient of area $i$. The standard industry practice is to set the bias
$\beta^{i}$ at the so-called \ac{AFRC}, which is defined as $\beta^{i} = D^{i} + 1/R^{i}$. The
integral of \ac{ACE} is used to construct the speed changer position feedback control signal
($\delta P^{i}_C$), i.e., $\delta P^{i}_C = -K^{i} x^{i}_{9}$, where $K^{i}$ is the feedback gain and
$\frac{dx^{i}_{9}}{dt}= ACE^{i}.$ 

The resulting state space model can be discretized and written in compact form as
\begin{equation} 
x(t_{k+1}) = A x(t_{k}) + B_1 u_{\anc}(t_{k}) + B_2 w(t_{k}). \label{eq:discrete_ss}
\end{equation}
Where $u_{\anc}=[\delta P^1_{\anc} \hdots \delta P^n_{\anc}]^T $ are the ancillary inputs, and the
exogenous inputs (i.e. disturbances or variations in demands) are denoted by $w=[\delta P_{D}^1
\hdots \delta P_{D}^n]^T$. We propose controller synthesis for the ancillary services, complementing the primary control of \ac{AGC}.

\subsection{\ac{MPC} for Ancillary Services}\label{sec:Predictive}
We require that $u_{\anc}$ be bounded and satisfies a maximum ramp constraint, i.e., $\underline{u}_{\anc} \leq u_{\anc}(t_k) \leq \overline{u}_{\anc}  \text{ with } \underline{u}_{\anc}
  > 0 \text{ and } 
| u_{\anc}(t_{k+1}) - u_{\anc}(t_k) | \leq \lambda, \text{for some } \lambda > 0$.
At each time step $k$, we thus solve the following problem:
\begin{align}\label{Eq:MPC}
& \underset{U_{\anc}(k)}{\min}& &  J(\ACE, U_{\anc})\\
& \text{s.t.}  & & x(t_{k+1}) = A x(t_{k}) + B_1 u_{\anc}(t_{k}) + B_2 w(t_{k}) \notag \\
& & & \underline{u}_{\anc} \leq u_{\anc}(t_{k+j}) \leq \overline{u}_{\anc} \notag \\
& & & \lvert u_{\anc}(t_{k+j+1}) - u_{\anc}(t_{k+j}) \rvert \leq \lambda  \notag 
\end{align}
\noindent where
$
U_{\anc}(k)= (u_{\anc}(t_k), u_{\anc}(t_{k+1}),~\ldots, u_{\anc}(t_{k+H}))
$
is the vector of inputs from $k$ to $k+H$ and $H$ is the prediction horizon. All the constraints of problem~\eqref{Eq:MPC} that depend on $j$ should hold for $j=0,1,\ldots,H-1$.\\

The cost function proposed in \cite{MaasoumySSP14} minimizes the $\ell_2$ norm of the $\ACE$ signal
in areas $i=1,\hdots,n$, by exploiting the ancillary service available in each area, while taking into
account the system dynamics and constraints. We propose to constrain the $\ACE$ signal to satisfy a
specified set of STL properties, while minimizing the ancillary service used by each area. Thus we defined
$J(ACE, U_\anc) = \| U_\anc \|_{\ell_2} = \sum_{i=1}^{2} \sum_{j=0}^{H-1} (U_\anc^i{[k+j]})^2$, and
an STL formula $\varphi$ which says that whenever $\abs{\ACE^i}$ is larger than 0.01, it should become less than 0.01 in less
than $\tau$ s. More precisely we used $\varphi= \G ( \varphi_t )$ with
\begin{equation}
  \begin{array}{l}
    \varphi_t =  \neg (\abs{\ACE^1} < .01 )) \Rightarrow ( \F_{[0, \tau]} (\abs{\ACE^1} < .01 ) \\
    ~~~\land
    ( \neg (\abs{\ACE^2} < .01 )) \Rightarrow ( \F_{[0, \tau]} (\abs{\ACE^2} < .01 )
  \end{array}\label{eq:stl}
\end{equation}
We encoded this formula and added the resulting constraints to the MPC problem as described in the previous
sections, and solved it for different values of $\tau$. Results are shown in Fig.~\ref{fig:ACE},
and demonstrate that the STL constraint is correctly enforced in the stabilization of the $\ACE$ signal.

\begin{figure}[t!] %htbp
\centering
\includegraphics[trim= 1cm 0 0 0, width=1.0\columnwidth]{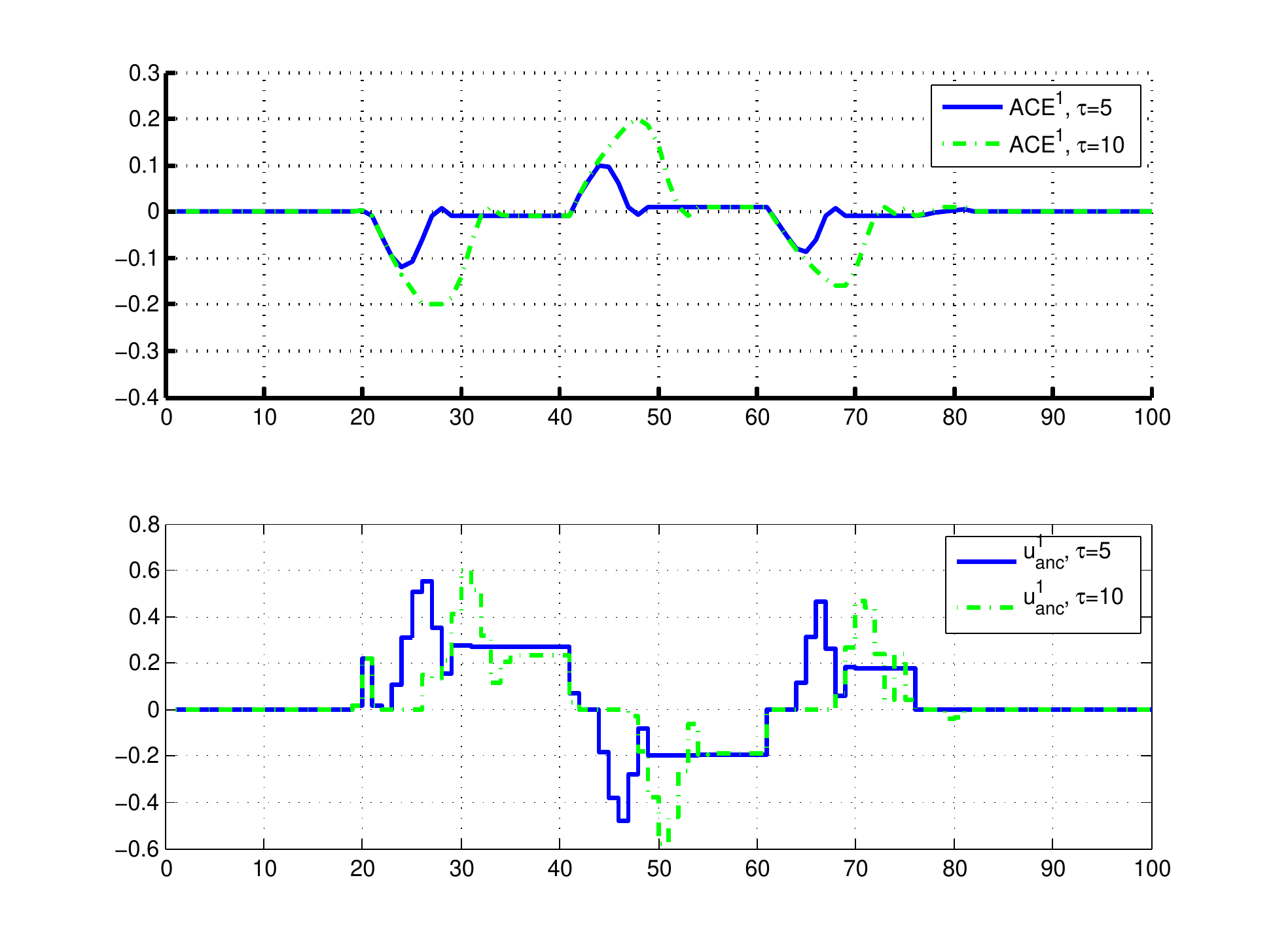}
\caption{Comparison of $\ACE^1$ (Area Control Error of Control Area 1) stabilization using $\varphi_t$ in (\ref{eq:stl}) with $\tau=5$ and
  $\tau=10$. The controller enforces the stabilization delay in both cases.}
\label{fig:ACE}
\end{figure}

%%% Local Variables:
%%% mode: latex
%%% TeX-master: "root_tac_part1"
%%% End:

\section{Related Work}
Receding horizon control for temporal logic has been considered before in the context of \ac{LTL}
~\cite{NokTM12}, where the authors propose a reactive synthesis scheme for specifications with GR(1)
goals.  The authors in \cite{GolL13} also propose an \ac{MPC} scheme for specifications in
synthetically co-safe \ac{LTL} -- our approach extends synthesis capabilities to a wider class of
temporal logic specifications.  In \cite{DingLB14}, the authors consider full \ac{LTL} but use an
automata-based approach, involving potentially expensive computations of a finite state abstraction
of the system and a Buchi automaton for the specification.  We circumvent these expensive operations
using a \ac{BMC} approach to synthesis.  In \cite{BemporadM99}, the authors present a model
predictive control scheme to stabilize mixed logical dynamical systems on desired reference
trajectories, while fulfilling propositional logic constraints and heuristic rules. A major
contribution of this work is to extend the constraint specification language for such systems to
\ac{STL} specifications, which allow expression of complex temporal properties including safety, liveness,
and response.

Our work extends the standard \ac{BMC} paradigm for finite discrete systems \cite{BiereCCZ99} to
\ac{STL}, which accommodates continuous systems. In BMC, discrete state sequences of a fixed length,
representing counterexamples or plans, are obtained as satisfying assignments to a Boolean
satisfiability (SAT) problem. The approach has been extended to hybrid systems, either by computing
a discrete abstraction of the system~\cite{GiorgettiPB05,jha-formats07} or by extending SAT solvers
to reason about linear inequalities \cite{AudemardBCS05,FranzleH05}. Similarly, \ac{MILP} encodings
inspired by \ac{BMC} have been used to generate trajectories for continuous systems with \ac{LTL}
specifications \cite{KaramanSF08,KwonA08,WolffTM14}, and for a restricted fragment of
\ac{MTL} without nested operators \cite{KaramanF08}.  While we draw much inspiration from these early efforts, 
ours is the first work to consider a \ac{BMC} approach to synthesis for full \ac{STL}.% specifications that allow arbitrary nesting of bounded operators, and top-level unbounded operators.
%For synthesis, we build on the recent success of encodings of Linear Temporal Logic (LTL) specifications as mixed integer-linear constraints \cite{WolffTM14}, based on linear encodings for bounded LTL model checking \cite{BiereHJLS06}.
%We extend the BMC-inspired \ac{MILP} encoding to \ac{STL}, and incorporate the resulting constraints into the optimization problem at each finite horizon of the \ac{MPC} computation.

%%% Local Variables:
%%% mode: plain-tex
%%% TeX-master: "root_tac_part1"
%%% End:

\section{Concluding Remarks}
The main contribution of this paper is a pair of bounded model checking style encodings for signal temporal logic specifications as mixed integer linear constraints. We showed how our encodings can be used to generate control for systems that must satisfy \ac{STL} properties, and additionally to ensure maximum robustness of satisfaction. Our formulation of the \ac{STL} synthesis problem can be used as part of existing controller synthesis frameworks to compute feasible and optimal controllers for cyber-physical systems. We presented experimental results for controller synthesis on simplified models of a smart micro-grid and HVAC system, and showed how the \ac{MPC} schemes in these examples can be framed in terms of synthesis from an STL specification, with simulation results illustrating the effectiveness of our proposed synthesis.

We have demonstrated the ability to synthesize control for systems on both the demand and supply sides of a smart grid. 
We view this as progress toward a contract-based framework for specifying and designing components of the smart grid and their interactions using \ac{STL} specifications. 
%Towards this goal, in PART II of this series we will describe
Future work includes a \emph{reactive synthesis} approach to synthesizing control inputs for systems operating in uncertain environments: we have already
demonstrated preliminary results in this direction in \cite{RamanDSMS15}. We will also further explore synthesis in an \ac{MPC} framework for unbounded STL properties. 
As mentioned in Section \ref{mpc_unbounded}, this is an easy extension of our approach for certain types of properties. 
Extending this to arbitrary properties has ties to online monitoring of \ac{STL} properties \cite{DeshmukhDGJJS15}, which is another direction of further exploration. %also discuss in that work. 

%Future work: summarize history, formula decomposition, ties to online monitoring for STL. Contract-based design (if the two systems satisfy STL specs, what can we guarantee when they interact?).

%%% Local Variables:
%%% mode: latex
%%% TeX-master: "root_tac_part1"
%%% End:

\bibliographystyle{IEEEtran}
\bibliography{tacrefs}

\vspace{-1cm}
\begin{IEEEbiography}[{\includegraphics[width=1in,clip,keepaspectratio]{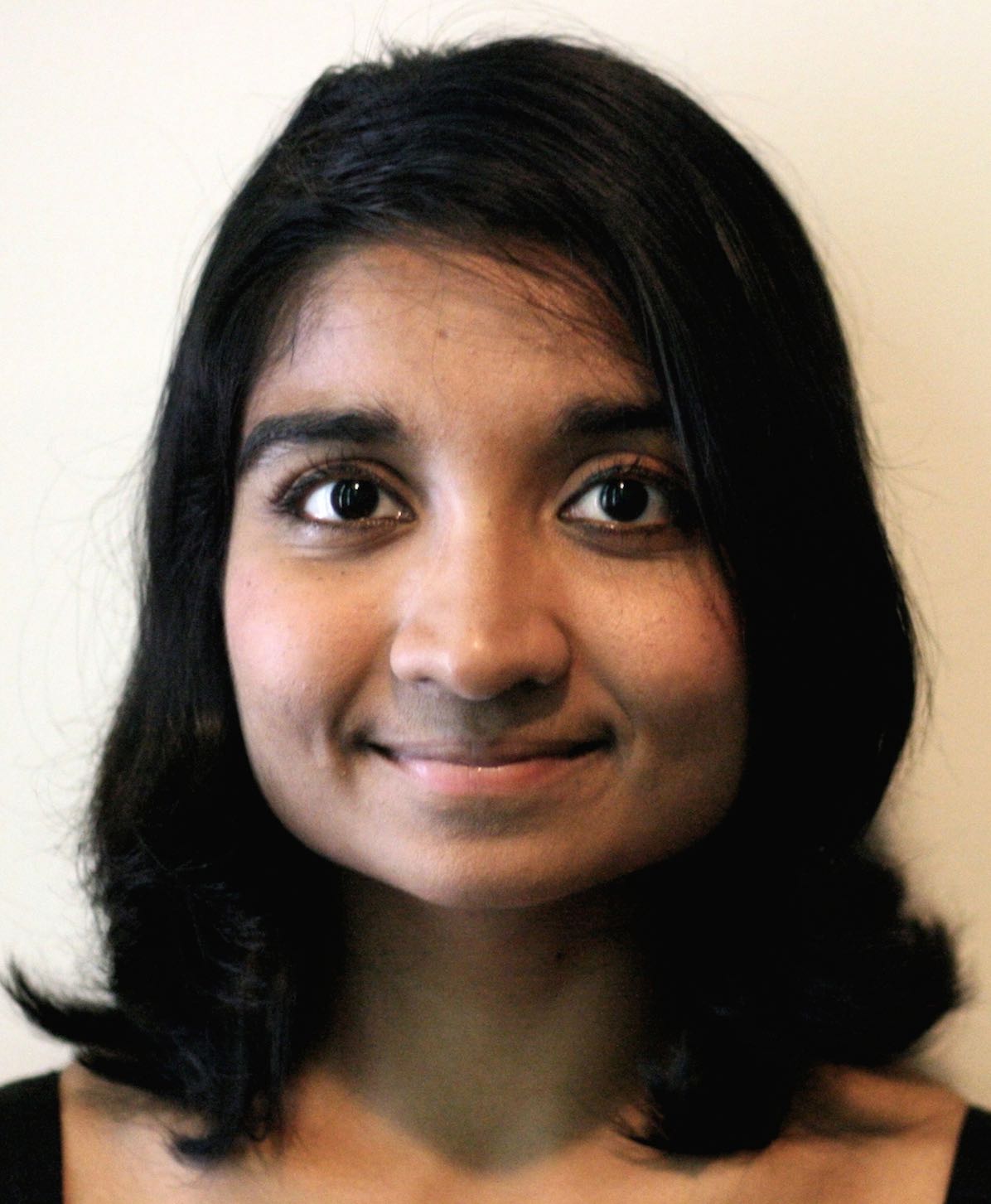}}]{Vasumathi Raman} received the B.A. degree in Computer Science and Mathematics from Wellesley College in 2007 and the M.S. and Ph.D. degrees in Computer Science from Cornell University in 2011 and 2013, respectively. She was a postdoctoral scholar in the Department of Computing \& Mathematical Sciences at the California Institute of Technology from 2013-2015, and is currently a Senior Scientist at the United Technologies Research Center in Berkeley, CA. Her research explores algorithmic methods for designing and controlling autonomous systems, guaranteeing correctness with respect to user-defined specifications.
\end{IEEEbiography}
\vspace{-0.2cm}
\begin{IEEEbiography}[{\includegraphics[width=1in,clip,keepaspectratio]{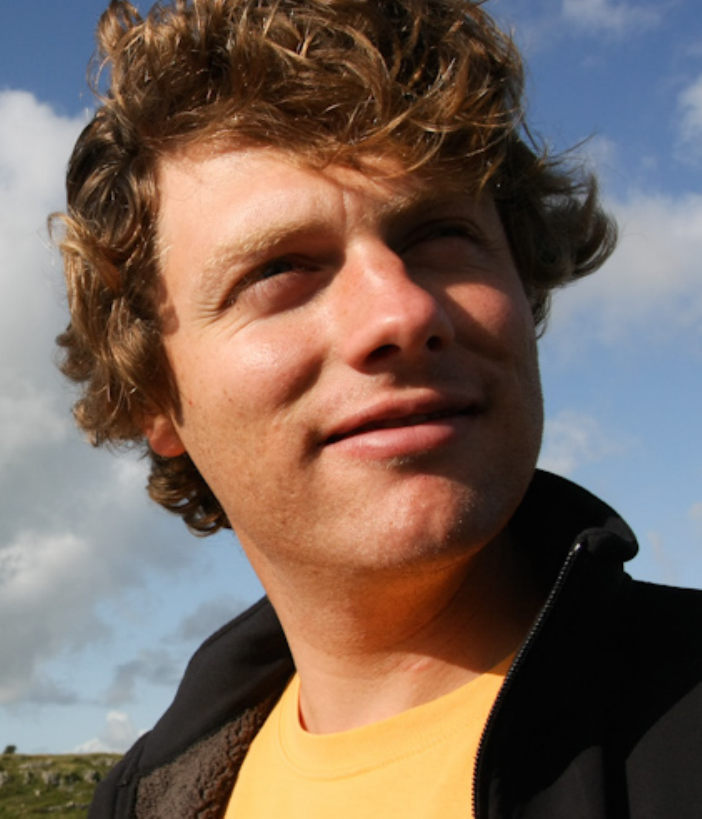}}]{Alexandre Donz\'{e}} is a research scientist at the University of California, Berkeley in the department of Electrical Engineering and Computer Science. He received his Ph.D. degree in Mathematics and Computer Science from the University of Joseph Fourier at Grenoble in 2007. He worked as a post-doctoral researcher at Carnegie Mellon University in 2008, and at Verimag in Grenoble from 2009 to 2012. His research interests are in simulation-based design and verification techniques using formal methods, Signal Temporal Logic (STL) with applications to cyber-physical systems and systems biology.
\end{IEEEbiography}
\vspace{-0.2cm}
\begin{IEEEbiography}[{\includegraphics[width=1in,clip,keepaspectratio]{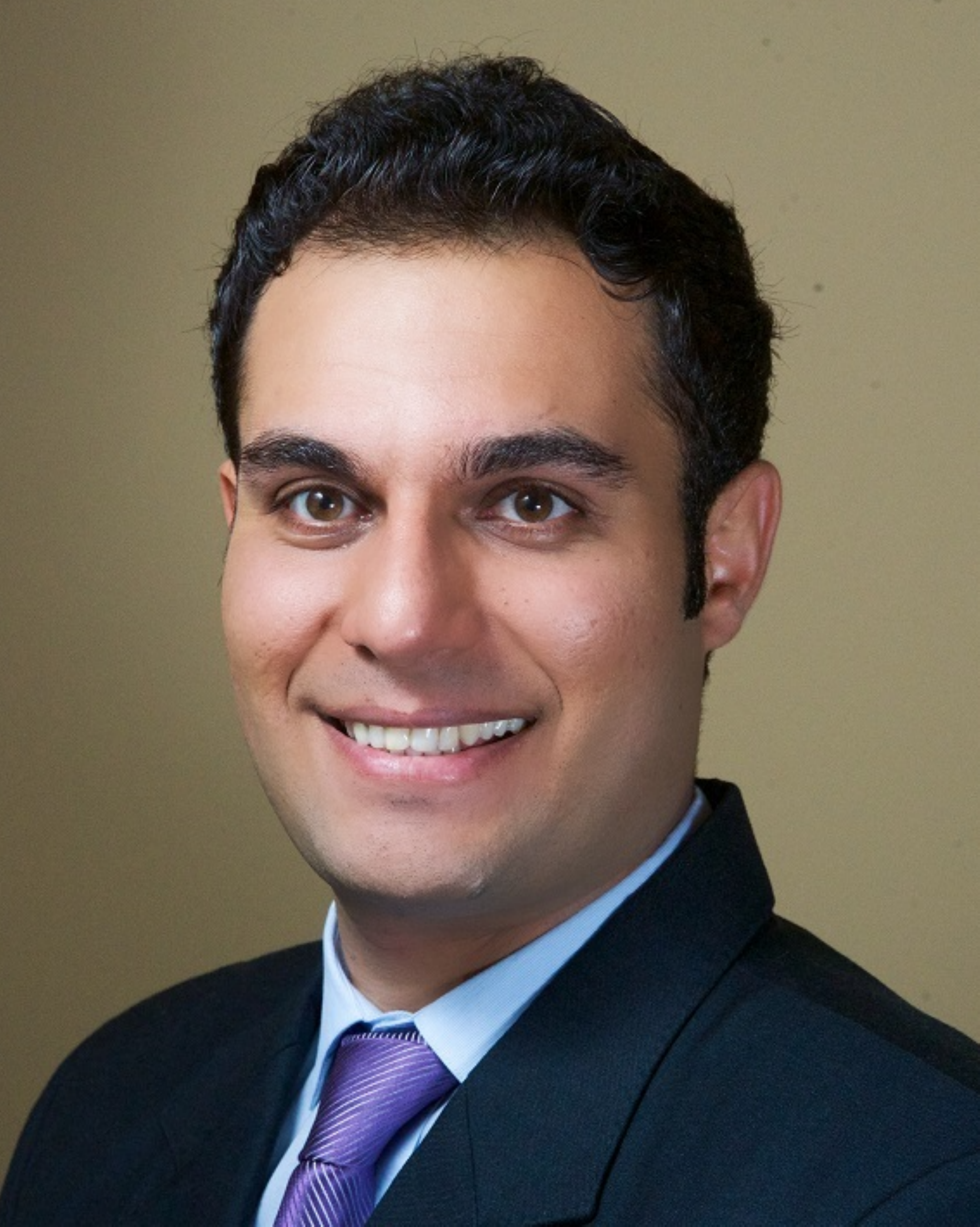}}]{Mehdi Maasoumy} received his PhD in  Electrical Engineering \& Computer Sciences in 2013 from UC Berkeley, where he was advised by Prof. Alberto Sangiovanni-Vincentelli. He received his MSc degree in 2010 from Mechanical Engineering Department at UC Berkeley with major in Controls and minor in Optimization with Designated Emphasis in Energy Systems and Technology (DEEST). He graduated form Sharif University of Technology in 2008 with a BSc degree in Mechanical Engineering.
\end{IEEEbiography}
\vspace{-0.2cm}
\begin{IEEEbiography}[{\includegraphics[width=1in,height=1.25in,clip,keepaspectratio]{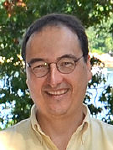}}]{Richard M. Murray} received the B.S. degree in Electrical Engineering from California Institute of Technology in 1985 and the M.S. and Ph.D. degrees in Electrical Engineering and Computer Sciences from the University of California, Berkeley, in 1988 and 1991, respectively. He is currently the Thomas E. and Doris Everhart Professor of Control \& Dynamical Systems and Bioengineering at Caltech. Murray's research is in the application of feedback and control to networked systems, with applications in biology and autonomy. Current projects include analysis and design biomolecular feedback circuits; specification, design and synthesis of networked control systems; and novel architectures for control using slow computing.
\end{IEEEbiography}
\vspace{-0.2cm}
\begin{IEEEbiography}[{\includegraphics[width=1in,height=1.25in,clip,keepaspectratio]{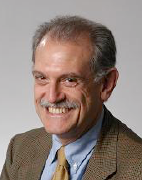}}]{Alberto Sangiovanni Vincentelli} holds the Edgar L. and Harold H. Buttner Chair of Electrical Engineering and Computer Sciences at the University of California at Berkeley. He has been on the Faculty of the Department since 1976. He obtained an electrical engineering and computer science degree (``Dottore in Ingegneria'') summa cum laude from the Politecnico di Milano, Milano, Italy in 1971. In 1980-1981, he spent a year as a Visiting Scientist at the Mathematical Sciences Department of the IBM T.J. Watson Research Center. In 1987, he was Visiting Professor at MIT. He has held a number of visiting professor positions at Italian Universities, including Politecnico di Torino, Universit\`{a} di Roma, La Sapienza, Universit\`{a} di Roma, Tor Vergata, Universit\`{a} di Pavia, Universit\`{a} di Pisa, Scuola di Sant'Anna. He was awarded the IEEE/RSE Wolfson James Clerk Maxwell Medal ``for groundbreaking contributions that have had an exceptional impact on the development of electronics and electrical engineering or related fields'' and the Kaufman Award of the Electronic Design Automation Council for ``pioneering contributions to EDA'' . He is Honorary Professor at Politecnico di Torino and has Honorary Doctorates from the University of Aalborg and KTH. He helped founding Cadence and Synopsys, the two leading companies in EDA. He is the author of over 850 papers, 18 books and 2 patents in the area of design tools and methodologies, large scale systems, embedded systems, hybrid systems and innovation.
\end{IEEEbiography}
\vspace{-0.2cm}
\begin{IEEEbiography}[{\includegraphics[width=1in,height=1.25in,clip,keepaspectratio]{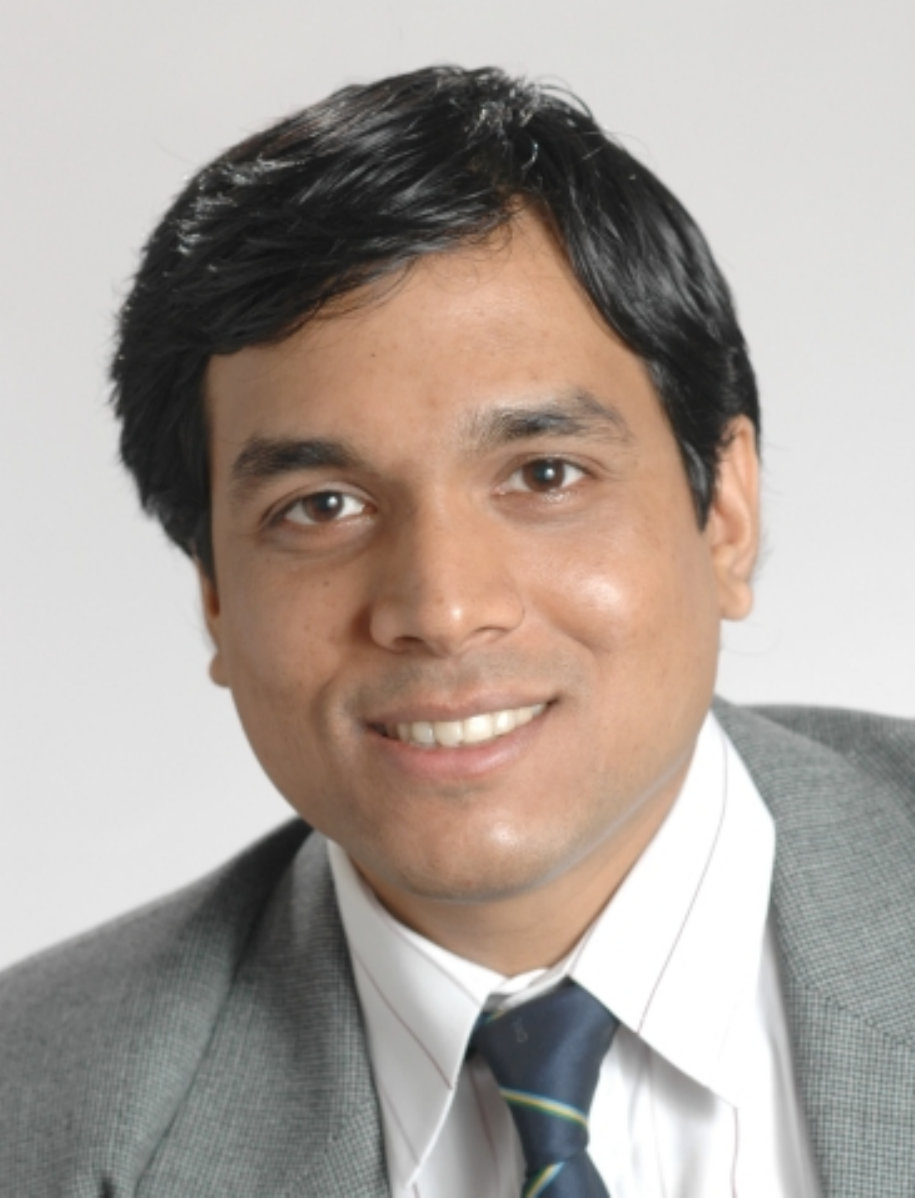}}]{Sanjit A. Seshia} received the B.Tech. degree in Computer Science and Engineering from the Indian Institute of Technology, Bombay in 1998, and the M.S. and Ph.D. degrees in Computer
Science from Carnegie Mellon University in 2000 and 2005 respectively. He is currently an Associate Professor in the Department of Electrical Engineering
and Computer Sciences at the University of California, Berkeley. His research interests are in dependable computing and computational logic, with
a current focus on applying automated formal methods to embedded and cyber-physical systems, electronic design automation, computer security, and synthetic biology. He has served as an Associate Editor of the IEEE Transactions on Computer-Aided Design of Integrated Circuits and Systems. His awards and honors include a Presidential Early Career Award for Scientists and Engineers (PECASE) from the White House, an Alfred P. Sloan Research Fellowship, the Prof. R. Narasimhan Lecture Award, and the School of Computer Science Distinguished Dissertation Award at Carnegie Mellon University.
\end{IEEEbiography}

\end{document}